\documentclass[10pt]{asme2ej}  
   
\usepackage{overpic,rotating,contour,footnote}                        
\usepackage{graphics,subfigure,enumerate}
\usepackage{graphicx,psfrag}
\usepackage{epsfig,wasysym}  
\usepackage{amssymb}
\usepackage[all]{xy}
\usepackage[color=aliceblue]{todonotes}
\usepackage{mathtools}
\usepackage{xcolor,hyperref}
\definecolor{aliceblue}{rgb}{0.94, 0.97, 1.0}

\usepackage{mathptmx}       
\usepackage{helvet,overpic}         
\usepackage{courier}        
\usepackage{graphicx,subfigure}        
\usepackage[bottom]{footmisc} 
\usepackage{epsfig} 
      
\def\Beweisende{\square}            
\def\BewEnde{\hfill{\Beweisende}}

\def\phm{{\hphantom{-}}} 

\def\phi{\varphi}

\def\RR{{\mathbb R}}

\def\NN{{\mathbb N}}

\def\dach#1{\widehat{#1}}

\def\dddot#1{\dot{\ddot{#1}}}
\def\ddddot#1{\ddot{\ddot{#1}}}

\def\Vkt#1{{\mathbf #1}} 

\newcommand{\kVkt}[1]{\widetilde{\Vkt #1}}  
\newcommand{\mVkt}[1]{\dach{\Vkt #1}}      

\def\Area{\mbox{Area}}
\def\Vol{\mbox{Vol}}

\usepackage[T1]{fontenc}
\usepackage{xhfill}
\newcommand{\ditto}[1][.4pt]{\xrfill{#1}~\textquotedbl~\xrfill{#1}}

\newtheorem{thm}{Theorem}
\newtheorem{lem}{Lemma}
\newtheorem{rmk}{Remark} 
\newtheorem{example}{Example}
\newtheorem{assumption}{Assumption}
\newtheorem{definition}{Definition}

\newtheorem{cor}{Corollary}

\title{Snappability and singularity-distance \\of pin-jointed body-bar frameworks}
\author{Georg Nawratil
    \affiliation{
	Institute of Discrete Mathematics and Geometry \& Center for Geometry and Computational Design\\
	TU Wien\\
	Wiedner Hauptstrasse 8-10/104, Vienna 1040, Austria\\
  Email: nawratil@geometrie.tuwien.ac.at
    }	  
}

\begin{document}

\maketitle    

\pagenumbering{arabic}

\begin{abstract}
{\it
It is well-known that there exist rigid frameworks  
whose physical models can snap between different realizations due to non-destructive elastic deformations of material. 
We present a method to measure this 
snapping capability based on the 
total elastic strain energy density of the framework by using the physical 
concept of Green-Lagrange strain. 
As this so-called {\it snappability} only depends on the intrinsic framework geometry, it enables a fair comparison of pin-jointed body-bar frameworks,  
thus it can serve engineers as a criterion within the design process of multistable mechanisms. 
Moreover, it turns out that the value obtained from this intrinsic pseudometric also 
gives the distance to the closest shaky configuration in the case of isostatic frameworks. 
Therefore it is suited for the computation of these singularity-distances for diverse mechanical devices. 
In more detail we study this problem for parallel manipulators of Stewart-Gough type.  
}
\end{abstract}

\section{Introduction}

In this paper we study frameworks composed of bars and bodies linked by pin-joints, 
which are rotational joints in the planar case and spherical joints in the spatial case. Note that all joints are assumed to be without clearance. 
A body is either a polyhedron or a polygonal panel\footnote{A polygonal panel can be seen as a body with coplanar vertices according to
\cite{kiraly}.}.  For both of these cases it is assumed that the body does not possess any unnecessary vertices; i.e.\ each of its 
vertices is pin-jointed. An additional assumption is that the inner graph of each 
body is globally rigid, where an inner graph is defined as follows: 

\begin{definition}\label{def:ingr}
Connect all vertices of the polyhedron (polygonal panel) by edges, which are either located on the boundary of the polyhedron (polygon) or in its interior. 
The resulting graph is called inner graph.
\end{definition}

Note that our studied class of geometric structures known as {\it pin-jointed body-bar frameworks} also 
contains hinge-jointed frameworks, as a hinge between two bodies can be replaced by two pin-joints.

By defining the combinatorial structure of the framework as well as the lengths of the bars and the shapes of the bodies, respectively,  
the intrinsic geometry of the framework is fixed. 
In general the assignment of the intrinsic metric does not uniquely determine the embedding of the framework into the Euclidean space, thus such a framework can have 
different incongruent realizations. 

A realization  is called a snapping realization if it is {\it close enough} 
to another incongruent realization such that the physical model can snap into this neighboring realization due to non-destructive elastic deformations of material. 
Shakiness can be seen as the limiting case where two realizations of a framework coincide; e.g.\ \cite{wohlhart,stachel_between}. 

The open problem in this context is the meaning of {\it closeness}, which is tackled in this article. In more detail, 
we present a method to measure the snapping capability (shortly called {\it snappability}) of a realization. 
The provided distance is of interest for practical applications, 
because it can be used in the early design phase of a framework to avoid snapping phenomena 
(e.g.\ engineering of truss structures) or to utilize them  (e.g.\ multistable mechanisms and materials). 
The latter approach has received much attention in the last few years within a wide field of applications; ranging 
from origami structures (e.g.\ \cite{liu2017,gillman,liu}) over mechanical metamaterials (e.g.\ \cite{haghpanah,shang,yang}) to metastructures (e.g.\ \cite{karpov2020}).

But the snappability also provides a distance to the next shaky configuration 
in the case of isostatic frameworks. Therefore this singularity-distance can also be used in the 
context of singularity-free path planing of robotic devices.

\subsection{Review and outline}

In two recent conference articles \cite{ark2020,ima2020} the author already started to investigate this topic. 
In \cite{ark2020} a first attempt towards the computation of the snappability of bar-joint frameworks was done 
based on the definition of Cauchy/Engineering strain. 
In \cite{ima2020} the author extended the approach to 
frameworks composed of bars and triangular panels. For this it was necessary to switch to the concept of 
Green-Lagrange strain, as the elastic strain energy of  triangular panels using Cauchy/Engineering strain is not invariant under rotations. 
In the articles \cite{ark2020,ima2020} the author restricted to planar examples; 
namely the trivial case of a triangular framework and the more sophisticated example of a pinned 3-legged planar 
parallel manipulator. 
In the paper at hand, we render the approach of \cite{ima2020} more precisely and generalize it to polygonal panels and polyhedra, respectively, and study some 
spatial frameworks, which already appear in existing literature on this topic reviewed next.

By the well-known technique of {\it deaveraging} (e.g.\ \cite{stachel_between}, \cite[page 1604]{schulze} and \cite{ivan}) snapping frameworks can be constructed in any dimension $\RR^d$. 
Moreover, for snapping bipartite frameworks in $\RR^d$ an explicit result in terms of confocal hyperquadrics is known 
(cf.\ \cite[page 112]{stachel_between} under consideration of \cite{stachel_palermo}). 
Most results are known for the dimension $d=3$, which are as follows:
There is a series of papers by Walter Wunderlich on snapping spatial structures 
(octahedra \cite{wunderlich_achtflach}, closed 4R loops \cite{wunderlich_gelenksviereck}, antiprisms \cite{wunderlich_antiprism}, 
icosahedra \cite{wunderlich_icosaeder1,wunderlich_icosaeder2}, dodecahedra \cite{wunderlich_dodekaeder}), 
which are reviewed in \cite{stachel_wunderlich}. In this context also the paper \cite{goldberg} 
should be cited, where {\it buckling polyhedral surfaces} and {\it Siamese dipyramids} are introduced. 
Snapping structures are also related to so-called {\it model flexors}\footnote{Mathematically these structures do not posses a 
continuous flexibility but due to free bendings without visible distortions of materials their physical models flex.} 
(cf.\ \cite{milka}) as in some cases the model flexibility can be reasoned by the snapping through different realizations. 
Examples for this phenomenon are the so-called {\it four-horn} \cite{schwabe} or the already mentioned {\it Siamese dipyramid}.
The latter are studied in more detail in \cite{gorkavyy}, especially how minor relative variations on the edge lengths 
produce significant relative variations in the spatial shape. The authors of \cite{gorkavyy} also suggested estimates to quantify these 
intrinsic and extrinsic variations. 
Recently, a more general approach for estimating these kinds of quantities for arbitrary bar-joint frameworks was presented in \cite{almost}, 
where inter alia also the Siamese dipyramid was studied as an example. 

Beside the above reviewed mathematical studies on snapping frameworks, there are also the following application driven approaches. 
Their snapping behavior is studied by 
\begin{enumerate}[(1)]
\item
 numerical simulations based on (a) finite element methods  \cite{haghpanah,shang,yang}  
or (b) force method approaches like \cite{guestII} or a generalized displacement control method \cite{liu2017,gillman,liu}, 
\item
 theoretical approaches based on the variation of the total potential energy \cite{yang,karpov2020,danso,klein}. 
\end{enumerate}
In contrast, our approach only relies on the  total strain energy of the structure  (i.e.\ no a priori assumptions 
on external loads have to be made) paving the way for the definition of the snappability, which only depends on the intrinsic framework geometry. 
Finally it should be noted, that a short review on approaches towards the computation of singular-distances 
is given in the section dealing with Stewart-Gough (SG) manipulators, which brings us straight to the outline of the paper.

After introducing notations and summarizing fundamentals of rigidity theory in Section \ref{sec:outline}, we present the underlying physical model of deformation in 
Section \ref{sec:phymodel}, which is used for building up the pseudometric on the space of intrinsic framework geometries 
in Section \ref{sec:pseudom}. Based on some theoretical considerations, we discuss the computation of the snappability and the singularity-distance in Sections \ref{sec:snap} 
and \ref{sec:dist}, respectively, and demonstrate the presented methods in two examples.  
Afterward we adopt our theoretical results for the singularity-distance computation of SG platforms in Section \ref{sec:SGplatform}, which is also closed by a practical example. 
Finally we conclude the paper in Section \ref{sec:conclusion}.  
Moreover, in Appendix \ref{appendix} the Siamese dipyramid and the four-horn are studied, and the obtained results are compared with existing literature.

\subsection{Notations and fundamentals of rigidity theory}\label{sec:outline}

A pin-jointed body-bar framework $G(\mathcal{K})$ consists of a knot set 
$\mathcal{K}=\Big\{V_{1},\ldots, V_r,$ $B_1^{d_1}(n_1),\ldots,$ $B_q^{d_q}(n_q)\Big\}$ 
and an 2-edge colored (green, red) graph $G$  on $\mathcal{K}$. 
A knot $B_i^{d_i}(n_i)$ represents a body, where $d_i\in\left\{2,3\right\}$ gives the additional information if the body is a polyhedron ($\Leftrightarrow$ $d_i=3$) 
or a polygonal panel ($\Leftrightarrow$ $d_i=2$). 
Without loss of generality we can assume that $d_1=\ldots=d_p=2$ and $d_{p+1}=\ldots =d_q=3$ for $1\leq p \leq q$. 
The number $n_i$ gives the number of vertices of the body $B_i$. A knot $V_i$ corresponds to rotational/spherical joint linking bars. 
A green edge connecting two knots corresponds to a bar. A red edge is only allowed to connect two bodies and represents a pin-joint.

Due to the assumed global rigidity of the inner graph of a body we can replace each body by a globally rigid bar-joint subframework 
according to  \cite[page 437]{kiraly}. 
Note that the combinatorial characterization of global rigidity is only known for $\RR^2$ (cf.\ \cite{jackson}), 
but still open for $\RR^3$ \cite[page 450]{kiraly}.  

\begin{rmk}\label{rem:1}
The completeness of an inner graph is a sufficient condition for global rigidity (cf.\ \cite{asimow}). 
This implies that the body $B_i^{d_i}(n_i)$ has to be convex as all $n_i(n_i-1)/2$ 
edges are in the interior of the polyhedron (polygon) or on its boundary.

Note that the globally rigid bar-joint subframework of a polygonal panel $B_i^2(n_i)$ is not infinitesimal rigid in $\RR^3$ for $n_i>3$ because every vertex can be 
infinitesimally flexed out of the plane spanned by the remaining $n_i-1$ vertices. 
\hfill $\diamond$
\end{rmk}

By replacing the bodies by globally rigid bar-joint subframeworks resulting from the inner graphs, we end up with a bar-joint framework $G_*(\mathcal{K}_*)$ 
which is equivalent to the given pin-jointed body-bar framework $G(\mathcal{K})$. 
This bar-joint framework $G_*(\mathcal{K}_*)$ can be used for defining the 
intrinsic geometry of the framework $G(\mathcal{K})$ in a mathematical rigorous way. For doing this, we introduce the following notation. 

By denoting the vertices of the body $B_i^{d_i}(n_i)$ by $V_{s_i+1}, \ldots ,V_{s_i+n_i}$ with 
$s_i=r+\sum_{j=1}^{i-1} n_j$ we get the set $\mathcal{K}_*=\left\{V_1, \ldots , V_w\right\}$ with 
$w=r+\sum_{j=1}^{q} n_j$. Moreover, we denote the edge connecting  $V_i$ to $V_j$ by $e_{ij}$ with $i<j$. 
Now we can fix the intrinsic metric of the framework $G_*(\mathcal{K}_*)$ (and therefore also of $G(\mathcal{K})$) by assigning a length $L_{ij}\in\RR_{>0}$ to each edge $e_{ij}$. 
Moreover, we collect all these lengths in the $b$-dimensional vector $\Vkt L=(\ldots, L_{ij},\ldots)^T$ of the space $\RR^b$ of 
intrinsic framework metrics, where $b$ gives the number of edges of the graph $G_*$.
Finally we collect the indices $ij$ of edges $e_{ij}$ of $G_*(\mathcal{K}_*)$ which correspond to 
green edges of $G(\mathcal{K})$ in the set $\mathcal{G}$. 

We denote a realization of the framework $G(\mathcal{K})$ and $G_*(\mathcal{K}_*)$ by $G(\Vkt V)$ and $G_*(\Vkt V)$, respectively, 
where the configuration of vertices $\Vkt V=(\Vkt v_1,\ldots ,\Vkt v_w)\in \RR^{wd}$ is composed by the coordinate vectors $\Vkt v_i=(x_i,y_i,z_i)^T$ for $d=3$ and 
$\Vkt v_i=(x_i,y_i)^T$ for $d=2$, respectively, 
of $V_i$ for $i=1,\ldots, w$.

Now we consider a realization  $G_*(\Vkt V)$ of the equivalent bar-joint framework $G_*(\mathcal{K}_*)$. 
In the rigidity community (e.g.\ \cite{connelly_book}) each edge $e_{ij}$ is assigned with a {\it stress (coefficient)} ${\omega}_{ij}\in\RR$.  
For every knot $V_i$ we can associate a so-called {\it equilibrium condition}
\begin{equation}\label{equilibrium}
\sum_{i<j\in N_i}{\omega}_{ij}(\Vkt v_i-\Vkt v_j) + \sum_{i>j\in N_i}{\omega}_{ji}(\Vkt v_i-\Vkt v_j)=\Vkt o
\end{equation}
where $\Vkt o$ denotes the $d$-dimensional zero-vector and $N_i$ the knot neighborhood of $V_i$; i.e. the index set of 
knots $\in\mathcal{K}_*$ connected with $V_i$ by bars. 
If for all $w$ knots this condition is fulfilled, then the $b$-dimensional stress-vector
${\omega}=(\ldots, {\omega}_{ij} , \ldots)^T$ is referred as {\it self-stress} (or {\it equilibrium stress}). \bigskip

\noindent
{\bf Algebraic approach to rigidity theory.}
The relation that two elements of the knot set are edge-connected can also be expressed algebraically. 
They are either quadratic constraints resulting from the squared distances of vertices (implied by a green edge) or they are 
linear conditions, which are stemming from the identification of vertices (implied by a red edge) or the elimination of 
isometries\footnote{This are 6 linear constraints for $d=3$ and 3 linear constraints for $d=2$.}.  
In total this results in a system of $n$ algebraic equations $c_1=0,\ldots ,c_n=0$  in $m$ unknowns\footnote{Note that for bar-joint frameworks this 
number equals $wd$, where $w$ is the number of vertices.},
which constitute an algebraic variety $A$.

If $A(c_1,\ldots ,c_n)$ is positive-dimensional then the framework is flexible; otherwise rigid. 
The framework is called minimally rigid (isostatic) if the removal of any algebraic constraint (resulting from an edge) will make the framework flexible. 
In this case $m=n$ has to hold. Rigid frameworks, which are not isostatic, are called {\it overbraced} or {\it overconstrained} ($n>m$). 
Note that there is also a combinatorial characterization of isostaticity for generic frameworks of dimension 2 according to the work of Laman \cite{laman}, 
but for dimension 3 this is still an open problem.

If $A(c_1,\ldots ,c_n)$ is zero-dimensional, then each real solution corresponds to a realization $G(\Vkt V_i)$ of the framework for $i=1,\ldots ,k$. 
If there is exactly one real solution, then the framework is called globally rigid. 
But one can also consider the complex solutions of the set of {\it realization equations} $c_1,\ldots ,c_n$ resulting in complex knot configurations $\Vkt V_i$ 
with $i=k+1,\ldots ,k+2f$ and $f\in\NN^*$ as they always appear in pairs. According to \cite{pak} they imply 
complex realizations $G(\Vkt V_i)$. 

We can compute in a realization the tangent-hyperplane to each of the hypersurfaces $c_i=0$ in $\RR^m$ for $i=1,\ldots ,n$. Note that this is always possible as 
all hypersurfaces are either hyperplanes or regular hyperquadrics.  
The normal vectors of these tangent-hyperplanes constitute the columns of a  $m\times n$ matrix  $\Vkt R_{G(\Vkt V)}$, which is also known as {\it rigidity matrix} 
of the realization $G(\Vkt V)$.  
If its rank is $m$ then the realization is infinitesimal rigid otherwise it is infinitesimal flexible; i.e. the hyperplanes 
have a positive-dimensional affine subspace in common. Therefore the intersection multiplicity of the $n$ hypersurfaces is at least two in a 
shaky realization. As a consequence shakiness (of order one\footnote{Each additional coinciding realization 
raises the order of the infinitesimal flexibility by one \cite{wohlhart}.}) can also be seen as the limiting case where 
two realizations of a framework coincide \cite{wohlhart,stachel_between,stachel_wunderlich}. 

Clearly, by using the rank condition $rk(\Vkt R_{G(\Vkt V)})<m$ one can also characterize all shaky realizations $G(\Vkt V)$ algebraically by  
the affine variety $A(J)$ -- which is referred as shakiness variety -- where $J$ denotes the ideal generated by all minors of $\Vkt R_{G(\Vkt V)}$ of order $m\times m$.
Let us assume that the polynomials $g_1,\ldots, g_{\gamma}$ form the Gr\"obner basis of the ideal $J$. 
Note that for minimally rigid framework $\gamma=1$ holds, where the infinitesimal flexibility is given by $g_1:\,\, \det(\Vkt R_{G(\Vkt V)})=0$.
Another approach towards this so-called {\it pure condition} in terms of brackets is given in \cite{white}.

\section{Physical model of deformation}\label{sec:phymodel}

The snappability index presented in this paper is based on the physical model of deformation relying on the concept of Green-Lagrange  strain,  
which is reduced to its geometric core by eliminating the influence of material properties. 
In order to do so, we make the following assumption.

\begin{assumption}\label{ass:1}
All bars and bodies of the framework are made of the same homogeneous isotropic material, which 
is non-auxetic; i.e.\ the Poisson ratio $\nu\in[0,1/2]$, and has a positive Young modulus $E>0$. 
\end{assumption}

\subsection{The relation between stress and strain}\label{sec:hook}
Due to the fact that the elastic  deformations  during the process of snapping are expected to be small, we can apply Hooke's law. 
As a consequence, the relation between applied stresses and resulting strains is a linear one, which can be given 
for the spatial case by
\begin{equation}\label{eq:stress_strain_spatial}
\underbrace{\begin{pmatrix}
\varepsilon_x \\
\varepsilon_y \\
\varepsilon_z \\
\gamma_{xy} \\
\gamma_{xz} \\
\gamma_{yz} 
\end{pmatrix}}_{\Vkt e_3}=\underbrace{\frac{1}{E}
\begin{pmatrix}
1 & -\nu & -\nu & 0 & 0 & 0 \\
-\nu & 1 & -\nu & 0 & 0 & 0 \\
-\nu & -\nu & 1& 0 & 0 & 0 \\
0 & 0 & 0 & 2(1+\nu) & 0 & 0 \\
0 & 0 & 0 & 0 & 2(1+\nu) & 0 \\
0 & 0 & 0 & 0 & 0 & 2(1+\nu) 
\end{pmatrix}}_{=:\Vkt D_3(\nu)}
\begin{pmatrix}
\delta_x \\
\delta_y \\
\delta_z \\
\tau_{xy} \\
\tau_{xz} \\
\tau_{yz} 
\end{pmatrix}
\end{equation}
where $\delta_i$ denotes the normal stress in $i$-direction and $\varepsilon_i$ its corresponding normal strain with $i\in\left\{x,y,z\right\}$. 
Moreover, $\tau_{ij}$ denotes the shear stress in the $ij$-plane and $\gamma_{ij}$ its corresponding shear strain with $i\neq j$ and  $i,j\in\left\{x,y,z\right\}$. 

In the case of planar stress ($xy$-plane) the shear stresses $\tau_{xz}$ and $\tau_{yz}$ are zero as well as the normal stress $\delta_z$. 
Then Eq.\ (\ref{eq:stress_strain_spatial}) simplifies to:
\begin{equation}\label{eq:stress_strain_planar}
\underbrace{\begin{pmatrix}
\varepsilon_x \\
\varepsilon_y \\
\gamma_{xy} 
\end{pmatrix}}_{\Vkt e_2}=\underbrace{\frac{1}{E}
\begin{pmatrix}
1 & -\nu  & 0 &  \\
-\nu & 1  & 0 &  \\
0 & 0 &  2(1+\nu) 
\end{pmatrix}}_{=:\Vkt D_2(\nu)}
\begin{pmatrix}
\delta_x \\
\delta_y \\
\tau_{xy} 
\end{pmatrix}.
\end{equation}
In the case of a bar (in $x$-direction) the relation reduces to $\varepsilon_x=D_1\delta_x$ with $D_1:=\tfrac{1}{E}$. 

For the later done computation of the elastic strain energies we need the inverse relations, which map the strains to the stresses. 
This can be obtained by inverting $D_1$, $\Vkt D_2$ and $\Vkt D_3$, respectively.   
$D_1$ does not depend on Poisson's ratio $\nu$ and its inverse reads as $D_1^{-1}=E$. 

As $\Vkt D_2$  is regular for all possible Poisson ratios $\nu\in[0,1/2]$, we can always compute 
\begin{equation}\label{eq:d2}
\Vkt D_2^{-1}(\nu)  =E
\begin{pmatrix}
\tfrac{1}{1-\nu^2} & \tfrac{\nu}{1-\nu^2}  & 0 &  \\
\tfrac{\nu}{1-\nu^2} & \tfrac{1}{1-\nu^2}  & 0 &  \\
0 & 0 &  \tfrac{1}{2(1+\nu)} 
\end{pmatrix} \quad\text{for}\quad 0\leq\nu\leq\frac{1}{2}.
\end{equation} 
For the spatial case, $\Vkt D_3^{-1}(\nu)$ is only not defined if $\nu$ equals the upper border of $\frac{1}{2}$, thus we get:
\begin{equation}\label{eq:d3}
\Vkt D_3^{-1}(\nu)  =E
 \begin{pmatrix}
\tfrac{\nu-1}{2\nu^2+\nu-1} & \tfrac{-\nu}{2\nu^2+\nu-1} & \tfrac{-\nu}{2\nu^2+\nu-1} & 0 & 0 & 0 \\
\tfrac{-\nu}{2\nu^2+\nu-1} & \tfrac{\nu-1}{2\nu^2+\nu-1} & \tfrac{-\nu}{2\nu^2+\nu-1} & 0 & 0 & 0 \\
\tfrac{-\nu}{2\nu^2+\nu-1} & \tfrac{-\nu}{2\nu^2+\nu-1} & \tfrac{\nu-1}{2\nu^2+\nu-1} & 0 & 0 & 0 \\
0 & 0 & 0 & \tfrac{1}{2(1+\nu)} & 0 & 0 \\
0 & 0 & 0 & 0 & \tfrac{1}{2(1+\nu)} & 0 \\
0 & 0 & 0 & 0 & 0 & \tfrac{1}{2(1+\nu)} 
\end{pmatrix}
\quad\text{for}\quad  0\leq\nu<\frac{1}{2}. 
\end{equation} 
For $\nu=\tfrac{1}{2}$ we compute the Moore-Penrose pseudo inverse  of $\Vkt D_3$  which yields:
\begin{equation}\label{eq:d3_pseudo}
\Vkt D_3^{-1}(\tfrac{1}{2}) =E
 \begin{pmatrix}
\phm\tfrac{4}{9} & -\tfrac{2}{9}  & -\tfrac{2}{9} & 0 & 0 & 0 \\
-\tfrac{2}{9} & \phm\tfrac{4}{9}  & -\tfrac{2}{9} & 0 & 0 & 0 \\
-\tfrac{2}{9} & -\tfrac{2}{9} & \phm\tfrac{4}{9}  & 0 & 0 & 0 \\
0 & 0 & 0 & \tfrac{1}{3} & 0 & 0 \\
0 & 0 & 0 & 0 & \tfrac{1}{3} & 0 \\
0 & 0 & 0 & 0 & 0 & \tfrac{1}{3} 
\end{pmatrix}.
\end{equation}


\subsection{Strain energy according to Green-Lagrange}\label{sec:GLstrain}

The study of the deformation of a polyhedron is based on the deformation of tetrahedra, which also play a central role in the stress
analysis within the finite element method  (e.g.\ see \cite[Chapter 6]{logan}). 
The strain computation for 3-simplices  according to Green-Lagrange is outlined next (e.g.\ see \cite[Section 2.4.2]{reddy}).

Let $V_a, V_b, V_c, V_d$ denote the vertices of the tetrahedron in the given undeformed configuration 
and  $V'_a, V'_b, V'_c, V'_d$ in the deformed one. 
Then there exists a uniquely defined $3\times 3$ matrix $\Vkt A$ which has the property
\begin{equation}\label{def:affine}
\Vkt A(\mVkt v_b-\mVkt v_a)=\mVkt v'_b-\mVkt v'_a, \qquad 
\Vkt A(\mVkt v_c-\mVkt v_a)=\mVkt v'_c-\mVkt v'_a, \qquad 
\Vkt A(\mVkt v_d-\mVkt v_a)=\mVkt v'_d-\mVkt v'_a, 
\end{equation}
where $\mVkt v_i$ (resp.\ $\mVkt v'_i$) is a 3-dimensional vector of $V_i$ (resp.\  $V'_i$) for $i\in\left\{a,b,c,d\right\}$ 
with respect to a Cartesian frame $\mathcal F$ (resp.\  $\mathcal F'$) attached to the undeformed (resp.\ deformed) tetrahedron. 
The Cartesian frame $\mathcal F$ can always be chosen in a way that  
its origin equals $V_a$, the vertex $V_b$ is located on its positive $x$-axis and 
$V_c$ is located in the $xy$-plane with a positive $y$ coordinate; i.e.\  
\begin{equation}\label{eq:coordtet}
\mVkt v_a=(0,0,0)^T, \quad \mVkt v_b=(x_{b},0,0)^T, \quad \mVkt v_c=(x_{c},y_{c},0)^T, \quad \mVkt v_d=(x_{d},y_{d},z_{d})^T,
\end{equation}
with $x_{b}>0$ and $y_{c}>0$. 
Similar considerations can be done for the Cartesian frame $\mathcal F'$ with respect to the tetrahedron $V_a', V_b', V_c', V_d'$ 
ending up with exactly the same coordinatization as above but only primed. 
Then the normal strains and the shear strains can be computed as 
\begin{equation}\label{eq:equi}
\begin{pmatrix}
\varepsilon_{x} & \tfrac{\gamma_{xy}}{2} & \tfrac{\gamma_{xz}}{2} \\
\tfrac{\gamma_{xy}}{2} & \varepsilon_{y} & \tfrac{\gamma_{yz}}{2} \\
\tfrac{\gamma_{xz}}{2} & \tfrac{\gamma_{yz}}{2}  & \varepsilon_{z}
\end{pmatrix}=\frac{1}{2}\left(\Vkt A^T\Vkt A-\Vkt I\right).
\end{equation}
Reassembling these quantities in the vector $\Vkt e_3$ (cf.\ Eq.\ (\ref{eq:stress_strain_spatial})) 
the elastic strain energy of the deformation can be calculated as 
\begin{equation}\label{esetp3}
U_{abcd}=\Vol_{abcd}\tfrac{1}{2}\Vkt e_3^T\Vkt D^{-1}_3(\nu)\Vkt e_3
\end{equation}
where $\Vol_{abcd}$ denotes the volume of the undeformed tetrahedron and $\Vkt D^{-1}_3(\nu)$ the stress/strain matrix (constitutive matrix) 
from Eq.\  (\ref{eq:d3}) and  Eq.\  (\ref{eq:d3_pseudo}), respectively.

The same procedure can be done for the computation of the elastic strain energy $U_{abc}$ of a triangular panel with vertices $V_a$, $V_b$ and $V_c$, which is outlined in detail in \cite{ima2020}. 
As final formula we obtain in this case: 
\begin{equation}\label{esetp2}
U_{abc}=\Vol_{abc}\tfrac{1}{2}\Vkt e_2^T\Vkt D^{-1}_2(\nu)\Vkt e_2
\end{equation}
where  $\Vkt D^{-1}_2(\nu)$ denotes the stress/strain matrix from Eq.\  (\ref{eq:d2}) and  $\Vol_{abc}$ the volume of the undeformed panel, which can be computed as the product of the 
triangle's area $\Area_{abc}$ and the panel height $h_{abc}$.  
For a bar with end-points $V_a$ and $V_b$ we end up with the following simple expression: 
\begin{equation}\label{bar_GL}
U_{ab}=\frac{E\Vol_{ab}}{8L_{ab}^4} ({L'_{ab}}^2-L_{ab}^2)^2 
\end{equation}
where $\Vol_{ab}$ denotes the volume of the undeformed bar, which can be computed as the product of the length $L_{ab}$ of the undeformed bar 
and its cross-sectional area  $\Area_{ab}$. The deformed bar length is given by $L'_{ab}$.


\section{A pseudometric on the space of intrinsic framework metrics}\label{sec:pseudom}

In this section we set up a pseudometric on the space of intrinsic framework metrics, which is based on 
the total strain energy density of the framework, because in this way the distance is invariant under scaling (change of unit length). 
Moreover, it allows to compare pin-jointed body-bar frameworks differing in the number of knots, the combinatorial structure and intrinsic metric.

\subsection{The strain energy density of a framework}\label{sec:strainenergy}

We assume that the intrinsic metric of the framework $G(\mathcal{K})$ is given by the edge-length vector $\Vkt L=(\ldots, L_{ij},\ldots)^T\in \RR^b$ 
of the equivalent framework $G_*(\mathcal{K}_*)$. 
In the same way the intrinsic metric of the deformed framework is determined by  $\Vkt L'=(\ldots, L'_{ij},\ldots)^T\in \RR^b$. 
As the strain energy of a polyhedron (polygonal panel) depends on its tetrahedralization\footnote{Decomposition into a set of disjoint 
tetrahedra without adding new vertices.} (triangulation\footnote{Decomposition into a set of disjoint triangles without adding new vertices.})
we compute the strain energy over all tetrahedra of the polyhedron (triangles of the polygonal panel). 
To do so, we define the index set $\mathcal{C}_i$ containing all index 4-tuple $abcd$ (3-tuple $abc$) 
with $a<b<c<d$ (with $a<b<c$) of non-degenerated\footnote{The tetrahedron (triangle) does not degenerate into a plane (line). 
Note that $\#\mathcal{C}_i=(n_i-3)(n_i-2)(n_i-1)n_i/24$ ($\#\mathcal{C}_i=(n_i-2)(n_i-1)n_i/6$) holds if the polyhedron 
(polygonal panel) is strictly convex.} tetrahedra (triangles) within a polyhedron $B_i^3(n_i)$ (polygonal panel $B_i^2(n_i)$). 
Using this notation we can formulate the strain energy density within the next lemma.

\begin{lem}\label{basic}
The strain energy density of a pin-jointed body-bar framework given by
\begin{equation}\label{eq:u}
u(\Vkt L'):=\frac{\sum_{ab\in\mathcal{G}} U_{ab}(\Vkt L')+
\sum_{i=1}^p \text{Vol}(B_i^2)\left[\frac{\sum_{abc\in\mathcal{C}_i} U_{abc}(\Vkt L')}{\sum_{abc\in\mathcal{C}_i} \Vol_{abc}}  \right] + 
\sum_{j=p+1}^q \text{Vol}(B_j^3)\left[\frac{\sum_{abcd\in\mathcal{C}_j} U_{abcd}(\Vkt L')}{\sum_{abcd\in\mathcal{C}_j} \Vol_{abcd}} \right]}
{\sum_{ab\in\mathcal{G}} \Vol_{ab}+
\sum_{i=1}^p \text{Vol}(B_i^2) + 
\sum_{j=p+1}^q \text{Vol}(B_j^3)}
\end{equation}
is defined by the intrinsic metric $\Vkt L$ of the undeformed framework, the 
cross-sectional areas $\Area_{ab}$ of its bars, the panel heights $h_{abc}$ and the material constants $E$ and $\nu$. 
The argument of the density function is given by the intrinsic metric $\Vkt L'$ of the deformed framework. It 
is a fourth order polynomial with respect to the variables $L_{ij}'$ which only appear with even powers.   
\end{lem}

\noindent
Proof:
We prove this lemma by investigating each summand in the numerator of Eq.\ (\ref{eq:u}) for the stated properties. 
As the energy functions differ for bars, triangular panels and tetrahedra, we have to split up the proof into 
these three cases:
\begin{enumerate}[$\bullet$]
\item 
{\it Bar:} For bars this result follows directly from Eq.\ (\ref{bar_GL}). 
\item
{\it Triangular panel:} 
We choose a planar Cartesian frame $\mathcal F$ in a way that the coordinates of the triangle $V_a, V_b, V_c$ 
read as $\mVkt v_a=(0,0)^T$, $\mVkt v_b=(x_b,0)^T$ and $\mVkt v_c=(x_c,y_c)^T$ with 
\begin{equation}\label{eq:plancoord}
x_b=L_{ab}, \quad
x_c=\tfrac{L_{ab}^2+L_{ac}^2-L_{bc}^2}{2L_{ab}},\quad 
y_c=\tfrac{\sqrt{
(L_{ab} + L_{ac} + L_{bc})(L_{ab} - L_{ac} + L_{bc})(L_{ab} + L_{ac} - L_{bc})(L_{ac} + L_{bc} - L_{ab})}}{2L_{ab}}
\end{equation}
where the coordinate $y_c$ can have positive or negative sign for planar frameworks 
depending on the orientation of the triangle $V_a, V_b, V_c$.  For spatial frameworks one can always assume a positive sign.   
Similar considerations can be done for the planar Cartesian frame $\mathcal F'$ with respect to the triangle $V_a', V_b', V_c'$ 
ending up with exactly the same coordinatization as above but only primed. Inserting these coordinates of the 
six vectors $\mVkt v_a,\mVkt v_b,\mVkt v_c,\mVkt v_a',\mVkt v_b',\mVkt v_c'$ into Eq.\ (\ref{esetp2})
shows the result for triangular panels by taking into account that the area $\Area_{abc}$ of the triangle can be computed by Heron's formula. 
Note that the obtained expression is independent of the signs of the $y$-coordinates of $\mVkt v_c$ and $\mVkt v_c'$.
\item
{\it Tetrahedron:} 
We choose the same Cartesian frame $\mathcal F$ as in Section \ref{sec:GLstrain} which implies the coordinatization of the 
tetrahedron $V_a, V_b, V_c,V_d$ given in Eq.\ (\ref{eq:coordtet}) with $x_b$, $x_c$ and $y_c$ from Eq.\ (\ref{eq:plancoord}) and 
\begin{equation}
\begin{split}
x_d=&\tfrac{L_{ab}^2 + L_{ad}^2 - L_{bd}^2}{2L_{ab}}, \quad
y_d=\tfrac{L_{ab}^2L_{ac}^2 + L_{ab}^2L_{ad}^2 + L_{ab}^2L_{bc}^2 + L_{ab}^2L_{bd}^2 - 2L_{ab}^2L_{cd}^2 - L_{ac}^2L_{ad}^2 + L_{ac}^2L_{bd}^2 + L_{ad}^2L_{bc}^2 - L_{bc}^2L_{bd}^2 -L_{ab}^4} 
{2L_{ab}\sqrt{(L_{ab} + L_{ac} + L_{bc})(L_{ab} - L_{ac} + L_{bc})(L_{ab} + L_{ac} - L_{bc})(L_{ac} + L_{bc} - L_{ab})}}  \\
z_d=&{\scriptstyle (-L_{ab}^4L_{cd}^2 - L_{ab}^2L_{ac}^2L_{bc}^2 + L_{ab}^2L_{ac}^2L_{bd}^2 + L_{ab}^2L_{ac}^2L_{cd}^2 + L_{ab}^2L_{ad}^2L_{bc}^2 - L_{ab}^2L_{ad}^2L_{bd}^2 + L_{ab}^2L_{ad}^2L_{cd}^2 
+ L_{ab}^2L_{bc}^2L_{cd}^2 + L_{ab}^2L_{bd}^2L_{cd}^2 - L_{ab}^2L_{cd}^4 - L_{ac}^4L_{bd}^2 +} \\ 
&{\scriptstyle 
 L_{ac}^2L_{ad}^2L_{bc}^2 + L_{ac}^2L_{ad}^2L_{bd}^2 - L_{ac}^2L_{ad}^2L_{cd}^2 
+ L_{ac}^2L_{bc}^2L_{bd}^2 - L_{ac}^2L_{bd}^4 + L_{ac}^2L_{bd}^2L_{cd}^2 - L_{ad}^4L_{bc}^2 - L_{ad}^2L_{bc}^4 + L_{ad}^2L_{bc}^2L_{bd}^2 + L_{ad}^2L_{bc}^2L_{cd}^2 - L_{bc}^2L_{bd}^2L_{cd}^2
)^{1/2}}/ \\
&{\scriptstyle \sqrt{(L_{ab} + L_{ac} + L_{bc})(L_{ab} - L_{ac} + L_{bc})(L_{ab} + L_{ac} - L_{bc})(L_{ac} + L_{bc} - L_{ab})}}
\end{split}
\end{equation}
where the coordinate $z_d$ can have positive or negative sign depending on the orientation of the tetrahedron $V_a, V_b, V_c,V_d$. 
We get the same coordinatization for the tetrahedron $V'_a, V'_b, V'_c,V'_d$ as above but only primed. Inserting these eight vectors 
$\mVkt v_a,\mVkt v_b,\mVkt v_c, \mVkt v_d, \mVkt v_a',\mVkt v_b',\mVkt v_c', \mVkt v_d'$
into Eq.\ (\ref{esetp3}) under consideration that $\Vol_{abcd}$ can be computed by the Cayley-Menger determinant
shows the stated result. 
Note that the obtained expression is independent of the sign of the $z$-coordinate of $\mVkt v_d$ and $\mVkt v_d'$. 
Moreover, it should be mentioned that the obtained polynomial is homogenous of degree 4 in  $L_{ij}'$ for $\nu=1/2$.  \hfill $\BewEnde$
\end{enumerate}

\begin{rmk}\label{rmk:lemma1}
Concerning Lemma \ref{basic} the following should be noted:
\begin{enumerate}[$\star$]
\item
The expressions given in the square brackets of Eq.\ (\ref{eq:u}) can be seen as the mean densities 
of the polygonal panels and polyhedra, respectively.
\item 
Note that the height $h_{abc}$ of each triangular panel belonging to $B_i^2(n_i)$ equals the 
height $h_i$ of $B_i^2(n_i)$.   
\item
Due to Lemma \ref{basic} the formula for $u(\Vkt L')$ can be written in matrix formulation as 
$u(\Vkt Q')=\Vkt Q'^T \Vkt M \Vkt Q'$ where $\Vkt M$ is a symmetric $(b+1)$-matrix and $\Vkt Q':=(1, \ldots ,Q_{ij}', \ldots)^T$
is composed of the $b$ squared edge lengths $Q_{ij}':=L_{ij}'^2$ and the number $1$. \hfill $\diamond$
\end{enumerate}
\end{rmk}

\subsection{Definition of the pseudometric}

The pseudometric on the space $\RR^b$ of intrinsic framework metrics is defined within the next lemma:

\begin{lem}\label{pseudo}
The following function 
\begin{equation}\label{pseudo}
d:\,\,
\RR^b\times\RR^b\rightarrow \RR_{\geq 0} \quad\text{with}\quad
(\Vkt L',\Vkt L'')\mapsto d(\Vkt L',\Vkt L''):= \frac{|u(\Vkt L')-u(\Vkt L'')|}{E}
\end{equation}
is a pseudometric on the $b$-dimensional space of 
intrinsic framework metrics given by  $\Vkt L'$  and $\Vkt L''$, respectively. 
Moreover, the pseudometric does not depend on the choice of $E$. 
\end{lem}

\noindent
Proof: 
One has to check the axioms for a pseudometric
\begin{equation}\label{axioms}
(1)\,\,\, d(\Vkt L',\Vkt L'')\geq 0, \quad
(2)\,\,\, d(\Vkt L',\Vkt L')= 0, \quad
(3)\,\,\, d(\Vkt L',\Vkt L'')=d(\Vkt L'',\Vkt L'), \quad
(4)\,\,\, d(\Vkt L',\Vkt L''')\leq d(\Vkt L',\Vkt L'')+d(\Vkt L'',\Vkt L'''),
\end{equation}
which is a trivial task and remains to the reader. 

Due to Assumption \ref{ass:1}, Young's modulus $E$ factors out of the density $u(\Vkt L')$. Therefore it factors out of the numerator of the 
distance function and cancels with the numerator. \bigskip \hfill $\BewEnde$

From Lemma \ref{basic} it is clear that the pseudodistance of Eq.\ (\ref{pseudo}) does not only depend on the intrinsic metric $\Vkt L$ of the undeformed framework 
but also on the cross-sectional areas of the bars and the heights of the panels, which are needed for the computation of their volumes. 
In the following section we fix these parameters by relating them to the intrinsic geometry of the undeformed framework.

\subsubsection{Geometric motivated volumetric dimensioning of bars and panels} \label{sec:dim}

As mentioned in Section \ref{sec:outline} each pin-jointed body-bar framework $G(\mathcal{K})$ can be replaced by an equivalent bar-joint framework $G_*(\mathcal{K}_*)$. 
In order to ensure a fair comparability of both frameworks, we came up with the following assumption.
\begin{assumption}
The frameworks $G(\mathcal{K})$ and $G_*(\mathcal{K}_*)$ have the same volume; i.e.\  they are built from the same amount of material.  
Moreover, we assume that all bars have the same cross-sectional area noted by $\Area_{\diameter}$. 
\end{assumption}
We start with the volumes of the polyhedra $\text{Vol}\left(B_j^{3}(n_j)\right)$, which are already determined by $\Vkt L$, 
 and compute $\Area_{\diameter}$ as 
\begin{equation}\label{eq:A}
\Area_{\diameter}:=\frac{\sum_{j=p+1}^q \text{Vol}\left(B_j^{3}(n_j)\right)}{\sum_{j=p+1}^q \sum_{ab\in\mathcal{I}_j}W_{ab}L_{ab}}
\end{equation}
where $\mathcal{I}_j$ is the index set of all pairs of vertices belonging to an edge of the inner graph of the polyhedron $B_j^{3}(n_j)$. 
Moreover, the weight factor $W_{ab}$ is one over the number of bodies hinged along the corresponding bar\footnote{If the bar does not hinge bodies, then the weight factor is one.}.

Now  having $\Area_{\diameter}$ one can also compute the height $h_i$ of the polygonal panel $B_i^{2}(n_i)$ over the bar-joint subframework 
equivalent to $B_i^{2}(n_i)$ as:
\begin{equation}
h_i:=\frac{\Area_{\diameter}\sum_{ab\in\mathcal{I}_i}W_{ab}L_{ab}}{\text{Area}\left(B_i^{2}(n_i)\right)}
\end{equation}
where $\mathcal{I}_i$ is the index set of all pairs of vertices belonging to an edge of the inner graph of the polygonal panel $B_i^{2}(n_i)$. 
In the case that the framework does not contain any polyhedra, then 
we can compute $h_i$ in the same way but depending on the unknown $\Area_{\diameter}$. 
In this case it can easily be seen that $\Area_{\diameter}$ factors out in the numerator as well as in the denominator of Eq.\ (\ref{eq:u}).  
Therefore Eq.\ (\ref{pseudo}) does not depend on $\Area_{\diameter}$. This also holds if the given framework 
is already a bar-joint framework.

\subsubsection{Geometric motivated choice of Poisson's ratio}\label{sec:pratio}

Under consideration of Section \ref{sec:dim} the pseudometric $d$ only depends on Poisson's ratio $\nu$ 
beside the intrinsic metric $\Vkt L$ of the undeformed framework. 
From the geometric point of view the most satisfying choice is $\nu=1/2$ as in this case the framework deforms at constant volume, which is also known 
as an  {\it isochoric} deformation. 
This does not pose any problems for a bar or triangular panel, as one can always define the cross-sectional  area $\Area'_{ab}$  of the deformed bar by $\Vol_{ab}/L'_{ab}$ and 
a height $h'_{abc}$ of the deformed panel by $\Vol_{abc}/\Area'_{abc}$, respectively, where $\Area'_{abc}$ is the area of the deformed triangle $V'_a$, $V'_b$, $V'_c$. 
But for each tetrahedron we get the additional condition that $\Vol_{abcd}=\Vol'_{abcd}$ holds, which fits very well into our theory for the following reason: 
The Moore-Penrose pseudo inverse of $\Vkt D_3^{-1}(\tfrac{1}{2})$ of Eq.\ (\ref{eq:d3_pseudo}) would imply that there 
exists a 1-dimensional set of strains $(\alpha,\alpha,\alpha,0,0,0)$ with $\alpha \in\RR$ yielding zero stresses, 
which cannot be the case\footnote{In fact the converse is true, that for $\nu=\tfrac{1}{2}$ there 
is no strain for uniform stresses.}. 
Plugging the entries of this strain vector into Eq.\ (\ref{eq:equi}) shows that it results from an 
equiform motion (Euclidean motion plus a scaling). But an equiform motion keeping the volume fixed has to be an Euclidean motion ($\Leftrightarrow$ $\alpha=0$), 
which does not imply any stress. Hence, the condition $\Vol_{abcd}=\Vol'_{abcd}$ resolves also the problem arising from the singularity of the 
stress/strain matrix given in Eq.\ (\ref{eq:d3_pseudo}).
Therefore one is only allowed to compute the pseudodistance of Eq.\ (\ref{pseudo}) for $\nu=1/2$ if 
$\Vol_{abcd}=\Vol'_{abcd}$ holds for all $abcd\in\mathcal{C}_j$ for $j=p+1,\ldots, q$.

Clearly, theoretically one can also use another Poisson ratio $0\leq\nu<1/2$ but in this paper we focus on the 
more sophisticated problem assuming constant volume under the deformation. A consequence of the constant volume deformation is that Eq.\ (\ref{eq:u}) 
cannot only be seen as the energy per volume, which has to be applied to the given framework to reach the deformed configuration but as the energy 
per volume which is stored in the deformed framework.

\begin{rmk}\label{rem:nu}
If one does not want to use the Poisson ratio $\nu=1/2$, then one is confronted to make a choice within the interval $[0;1/2[$. 
One can circumvent the arbitrariness in the choice by determining $\nu$ within a constrained optimization ($0\leq\nu<1/2$) in such a way that 
the distance of Eq.\ (\ref{pseudo}) is minimal. The disadvantage of this approach is that the triangular inequality 
of Eq.\ (\ref{axioms}) cannot longer be guaranteed thus the pseudometric degenerates to a so-called premetric. \hfill $\diamond$ 
\end{rmk}

Finally, it should be noted that the results of the next sections are general ones; i.e. the assumptions done in Section 3.2.1 and 3.2.2 are not necessary 
for their validity.


\section{Local and global snappability}\label{sec:snap}

If we want to compute the distance between $\Vkt L$ and $\Vkt L'$
then the pseudometric $d(\Vkt L,\Vkt L')$ simplifies to $u(\Vkt L')/E$ as this function is positive-definite, which is clear 
from the underlying physical interpretation but one can also prove this mathematically by decomposing it into a 
sum of squares (see e.g.\ \cite{reznick}).

As we can replace $L'_{ij}$ in $u(\Vkt L')$ by $\|\Vkt v'_i-\Vkt v'_j\|$ the function $u$ can be computed in dependence of $\Vkt V'$; i.e.\ $u(\Vkt V')$.

\begin{thm}\label{thm:critic}
For $0\leq\nu<1/2$ 
the critical points of the total elastic strain energy density $u(\Vkt V')$ of a pin-jointed body-bar framework 
correspond to realizations $G_*(\Vkt V')$ of the equivalent bar-joint framework that are either undeformed or 
deformed with a non-zero self-stress. This also holds for $\nu=1/2$ under the side conditions of constant tetrahedral volumes.
\end{thm}

\noindent
Proof: 
The system of equations characterizing the critical points of $u(\Vkt V')$ reads as follows:
\begin{equation}\label{defpartial}
\nabla_{\hspace{-0.5mm}i}\, u(\Vkt V')=\Vkt o  \quad \text{with} 
\quad
\begin{cases}
\nabla_{\hspace{-0.5mm}i}\, u(\Vkt V')=\left(\tfrac{\partial u}{\partial x'_{i}}, \tfrac{\partial u}{\partial y'_{i}}\right)  &\text{for $d=2$} \\
\nabla_{\hspace{-0.5mm}i}\, u(\Vkt V')=\left(\tfrac{\partial u}{\partial x'_{i}}, \tfrac{\partial u}{\partial y'_{i}}, \tfrac{\partial u}{\partial z'_{i}}\right)  &\text{for $d=3$} \\
\end{cases}
\end{equation}
with $i=1,\ldots, w$, where $(x'_i,y'_i)^T$ and  $(x'_i,y'_i,z'_i)^T$ is the coordinate vector of $\Vkt v'_i$ for the planar and spatial case, respectively.
Due to the sum rule for derivatives we only have to investigate $\nabla_{\hspace{-0.5mm}i}$ of the following three functions: 
$U_{ab}(\Vkt v_a',\Vkt v_b')$ of Eq.\ (\ref{bar_GL}), 
$U_{abc}(\Vkt v_a',\Vkt v_b',\Vkt v_c')$ given in Eq.\ (\ref{esetp2}) and $U_{abcd}(\Vkt v_a',\ldots ,\Vkt v_d')$ of Eq.\ (\ref{esetp3}), respectively.  
\begin{enumerate}[1.]
\item Due to  
$\nabla_{\hspace{-0.5mm}a}\, U_{ab}(\Vkt v_a',\Vkt v_b') = \tfrac{\Area_{ab}(L_{ab}'^2-L_{ab}^2)}{2L_{ab}^3}(\Vkt v'_a-\Vkt v'_b)$ Theorem \ref{thm:critic} is valid for frameworks, which only consist of bars, 
as $\nabla_{\hspace{-0.5mm}a}\, u(\Vkt V')$ can be written in the form of
Eq.\ (\ref{equilibrium}) with ${\omega}_{ab}=\tfrac{\Area_{ab}(L_{ab}'^2-L_{ab}^2)}{2L_{ab}^3}$. 
\item
If polygonal panels are involved we consider a representative triangular panel with vertices $V_a,V_b,V_c$ and compute  
$\nabla_{\hspace{-0.5mm}a}\, U_{abc}(\Vkt v_a',\Vkt v_b',\Vkt v_c')$, $\nabla_{\hspace{-0.5mm}b}\, U_{abc}(\Vkt v_a',\Vkt v_b',\Vkt v_c')$ and $\nabla_{\hspace{-0.5mm}c}\, U_{abc}(\Vkt v_a',\Vkt v_b',\Vkt v_c')$. 
Straight forward symbolic computations (e.g.\ using Maple) show that the following 
 system of equations
\begin{equation}\label{test1}
\begin{split}
{\omega}_{ab}(\Vkt v'_a-\Vkt v'_b) + {\omega}_{ac}(\Vkt v'_a-\Vkt v'_c) - \nabla_{\hspace{-0.5mm}a}\, U_{abc}(\Vkt v_a',\Vkt v_b',\Vkt v_c') &=\Vkt o \\
{\omega}_{ab}(\Vkt v'_b-\Vkt v'_a) + {\omega}_{bc}(\Vkt v'_b-\Vkt v'_c) - \nabla_{\hspace{-0.5mm}b}\, U_{abc}(\Vkt v_a',\Vkt v_b',\Vkt v_c') &=\Vkt o \\
{\omega}_{ac}(\Vkt v'_c-\Vkt v'_a) + {\omega}_{bc}(\Vkt v'_c-\Vkt v'_b) - \nabla_{\hspace{-0.5mm}c}\, U_{abc}(\Vkt v_a',\Vkt v_b',\Vkt v_c') &=\Vkt o 
\end{split}
\end{equation}
which is  overdetermined\footnote{As a triangle is planar, we get in total 6 equations in three unknowns from 
Eq.\ (\ref{test1}).}, 
has a unique solution for ${\omega}_{ab}$, ${\omega}_{ac}$ and ${\omega}_{bc}$ if $V_a', V_b', V_c'$ generate a triangle. If these points are collinear we even
get a positive dimensional solution set. Hence, one can replace $\nabla_{\hspace{-0.5mm}a}\, U_{abc}(\Vkt v_a',\Vkt v_b',\Vkt v_c')$ by a linear combination
${\omega}_{ab}(\Vkt v'_a-\Vkt v'_b) + {\omega}_{ac}(\Vkt v'_a-\Vkt v'_c)$ where the coefficients ${\omega}_{ab}$ and 
${\omega}_{ac}$ are compatible with the other equations of (\ref{test1}). 
\item
If bodies are involved we consider  a representative tetrahedron with vertices $V_a,V_b,V_c,V_d$ and compute  
$\nabla_{\hspace{-0.5mm}a}\, U_{abcd}(\Vkt v_a',\ldots,\Vkt v_d')$, $\nabla_{\hspace{-0.5mm}b}\, U_{abcd}(\Vkt v_a',\ldots,\Vkt v_d')$, $\nabla_{\hspace{-0.5mm}c}\, U_{abcd}(\Vkt v_a',\ldots,\Vkt v_d')$ 
and $\nabla_{\hspace{-0.5mm}d}\, U_{abcd}(\Vkt v_a',\ldots,\Vkt v_d')$.  
Now we are faced with the system of equations
\begin{equation}\label{test2}
\begin{split}
{\omega}_{ab}(\Vkt v'_a-\Vkt v'_b) + {\omega}_{ac}(\Vkt v'_a-\Vkt v'_c)  + {\omega}_{ad}(\Vkt v'_a-\Vkt v'_d) - \nabla_{\hspace{-0.5mm}a}\, U_{abcd}(\Vkt v_a',\ldots,\Vkt v_d') &=\Vkt o \\
{\omega}_{ab}(\Vkt v'_b-\Vkt v'_a) + {\omega}_{bc}(\Vkt v'_b-\Vkt v'_c)  + {\omega}_{bd}(\Vkt v'_b-\Vkt v'_d) - \nabla_{\hspace{-0.5mm}b}\, U_{abcd}(\Vkt v_a',\ldots,\Vkt v_d') &=\Vkt o \\
{\omega}_{ac}(\Vkt v'_c-\Vkt v'_a) + {\omega}_{bc}(\Vkt v'_c-\Vkt v'_b)  + {\omega}_{cd}(\Vkt v'_c-\Vkt v'_d) - \nabla_{\hspace{-0.5mm}c}\, U_{abcd}(\Vkt v_a',\ldots,\Vkt v_d') &=\Vkt o \\
{\omega}_{ad}(\Vkt v'_d-\Vkt v'_a) + {\omega}_{bd}(\Vkt v'_d-\Vkt v'_b)  + {\omega}_{cd}(\Vkt v'_d-\Vkt v'_c) - \nabla_{\hspace{-0.5mm}d}\, U_{abcd}(\Vkt v_a',\ldots,\Vkt v_d') &=\Vkt o 
\end{split}
\end{equation}
which is again overdetermined ($12$ equations in six unknowns ${\omega}_{ab},{\omega}_{ac},{\omega}_{ad},{\omega}_{bc},{\omega}_{bd}$ and 
${\omega}_{cd}$). Again direct computations show that there exists a unique solution if $V_a', V_b', V_c', V_d'$ span a 3-space; otherwise even a positive dimensional solution set exists.  
Thus one can substitute $\nabla_{\hspace{-0.5mm}a}\, U_{abcd}(\Vkt v_a',\ldots,\Vkt v_d')$ by a linear combination
${\omega}_{ab}(\Vkt v'_a-\Vkt v'_b) + {\omega}_{ac}(\Vkt v'_a-\Vkt v'_c)  + {\omega}_{ad}(\Vkt v'_a-\Vkt v'_d)$ where the coefficients ${\omega}_{ab}$, ${\omega}_{ac}$ and 
${\omega}_{ad}$ are compatible with the other equations of (\ref{test2}). 
\end{enumerate}
Summing up the results of the three items shows that $\nabla_{\hspace{-0.5mm}a}\, u(\Vkt V')$ can be written in the form of Eq.\ (\ref{equilibrium}) 
which proves the theorem for $\nu<1/2$.

For $\nu=1/2$ we have to compute the critical points of the Lagrange function
\begin{equation}\label{extendlagrange}
F(\Vkt V',\lambda)= u(\Vkt V') 
-{\lambda}_1f_1-\ldots -{\lambda}_{\phi}f_{\phi} \quad \text{with}\quad
\lambda:=({\lambda}_1,\ldots,{\lambda}_{\phi}), 
\end{equation}
where $f_1,\ldots,f_{\phi}$ are the isochoric constraints of the form ${\Vol'_{abcd}}^2-\Vol_{abcd}^2=0$ 
for $abcd\in\mathcal{C}_j$ for all $j=p+1,\ldots, q$; i.e\ $\phi=\sum_{j=p+1}^q\#\mathcal{C}_j$.
Due to the Cayley-Menger determinant we get 
the squared volume ${V'_{abcd}}^2$ of the tetrahedron spanned by $V'_a,\ldots ,V'_d$, 
as a polynomial in the squared distances of these vertices. Now one can replace in Eq.\ (\ref{test2}) the function $U_{abcd}(\Vkt v'_a,\ldots,\Vkt v'_d)$ 
by ${V'_{abcd}}^2(\Vkt v'_a,\ldots,\Vkt v'_d)$ and do the analogous computation ending up with the same conclusion. 
Therefore  $\nabla_{\hspace{-0.5mm}a}\,F(\Vkt V',\lambda)$ is again of the form of Eq.\ (\ref{equilibrium}) which proves the theorem for $\nu=1/2$. 
\hfill $\BewEnde$

\begin{rmk}\label{rmk:critic}
One can also ask for the critical points of the elastic strain energy density $u(\Vkt L')$; i.e.\ we have to consider the partial derivatives with respect to the 
edge lengths $L_{ij}'$. It can easily be seen that there is only one valid critical point namely $\Vkt L'=\Vkt L$ as all other solutions of the resulting system imply  
at least one edge of zero length. These invalid solutions can be avoided by considering $u(\Vkt Q')$ of Remark \ref{rmk:lemma1} and its partial derivatives with respect 
to $Q_{ij}'$ ending up in a linear system. \hfill $\diamond$
\end{rmk}

For the formulation of the next theorem we also need the notation of stability. 
A realization $G(\Vkt V')$ 
is called stable if it corresponds to a local minimum of the total elastic strain energy of the framework, which 
is also a minimum of the strain energy density $u(\Vkt V')$.

\begin{thm}\label{thm121}
If a pin-jointed body-bar framework snaps out of a stable realization $G(\Vkt V)$ 
by applying the minimum strain energy needed to it, then 
the corresponding deformation passes a realization $G(\Vkt V')$ at the maximum state of deformation, where the 
equivalent bar-joint framework $G_*(\Vkt V')$ has a non-zero self-stress. 
\end{thm}

\noindent
Proof: 
We think of $u$ as a graph function over the space $\RR^{wd}$ of knot configurations. 
In order to get out of the valley of the local minimum  $(\Vkt V, u(\Vkt V))$, which corresponds to the given stable realization $G(\Vkt V)$, 
with a minimum of energy needed, one has to pass a {\it saddle point} $(\Vkt V', u(\Vkt V'))$ of the graph, which corresponds to a realization $G(\Vkt V')$. 
As local extrema as well as saddle points of the graph function are given by the critical points of $u$ we can use 
Theorem \ref{thm:critic}, which implies that these points correspond with self-stressed realizations of the equivalent bar-joint framework. 
As  $u(\Vkt V')>0$ holds the realization $G_*(\Vkt V')$ is deformed which has to imply a non-zero self-stress; i.e.\ 
the stress-vector $\omega$ differs from the $b$-dimensional zero vector.
\hfill $\BewEnde$


\begin{cor}\label{rep:shaky}
If the equivalent bar-joint framework of Theorem \ref{thm121} is minimally rigid, then ``non-zero self-stress'' can be replaced by ``shakiness''. 
\end{cor}

\noindent
Proof: 
If the equivalent bar-joint framework is minimally rigid then the existence of a non-zero self-stress implies a rank defect of the square rigidity matrix (cf.\ end of Section \ref{sec:outline}), 
which results in an infinitesimal flexibility.   
\bigskip  \hfill $\BewEnde$

But also without the assumption of minimal rigidity used in Corollary \ref{rep:shaky}, one can give the  
following connection between shakiness and snapping.

\begin{thm}\label{thm122}
One can replace  ``non-zero self-stress'' by ``shakiness'' in Theorem \ref{thm121} 
if there exists a deformation such that the path from $G_*(\Vkt V)$ to $G_*(\Vkt V')$ is identical to 
the path  of $G_*(\Vkt V'')$ to $G_*(\Vkt V')$ in the space of intrinsic framework metrics $\RR^b$, where 
$G_*(\Vkt V)$ and $G_*(\Vkt V'')$ are not related by a direct isometry. 
\end{thm}

\noindent
Proof: 
Let us assume that there exists a path $(\Vkt V_t, u(\Vkt V_t))$
on the graph with parameter $t\in[0,1]$ such that for $t=0$ we are at  $(\Vkt V, u(\Vkt V))$ and for $t=1$ at the 
saddle $(\Vkt V', u(\Vkt V'))$. 
This deformation implies a path $\Vkt L_t$ in the space $\RR^b$ of intrinsic metrics with
$\Vkt L_t\big|_{t = 1}=\Vkt L'$ and $\Vkt L_t\big|_{t = 0}=\Vkt L$. 

But vice versa the path $\Vkt L_t$ corresponds to several
1-parametric deformations in $\RR^{d}$, where one of these deformations $(\overline{\Vkt V}_t, u(\overline{\Vkt V}_t))$ has to lead 
towards $(\Vkt V'', u(\Vkt V''))$ according to our assumption.
Moreover, tracking the realizations of the path $\Vkt L_t$ with $t\in[0,1]$ shows that in $G_*(\Vkt V')$ two realizations coincide, which implies that $G_*(\Vkt V')$ is shaky.  
\hfill $\BewEnde$

\begin{rmk}
Note that $G(\Vkt V'')$ can also be a complex realization. In this case the deformation $(\overline{\Vkt V}_t, u(\overline{\Vkt V}_t))$ 
towards $(\Vkt V'', u(\Vkt V''))$ get stuck on the border of reality. Therefore the snap ends up in a realization of the equivalent bar-joint framework, which is shaky  
as  a real solution of an algebraic set of equations can only change over into a complex one through a double root.  
\hfill $\diamond$
\end{rmk}

If the two realizations $G_*(\Vkt V)$ and $G_*(\Vkt V'')$ of Theorem \ref{thm122} are thought infinitesimal close to $G_*(\Vkt V')$ then we get 
the following characterization of shakiness:

\begin{cor}\label{cor:dif}
 $G_*(\Vkt V')$ is shaky, if three exist two instantaneous snapping deformations ($\neq$ 
infinitesimal isometric deformations) out of $G_*(\Vkt V')$ represented by two non-zero vectors in  $\RR^{wd}$ pointing into distinct directions,  
whose corresponding two vectors of instantaneous changes of the intrinsic metric in $\RR^b$ are identical. 
\end{cor}

\begin{rmk}
Note that within the set of pairs of vectors fulfilling Corollary \ref{cor:dif}, 
there exists at least one pair of oppositely directed vectors  in  $\RR^{wd}$, which both correspond to the zero-vector in $\RR^b$. 
\hfill $\diamond$
\end{rmk}

Based on Theorem \ref{thm121} we can evaluate the snappability of the pin-jointed body-bar framework $G(\mathcal{K})$ in the undeformed realization 
$G(\Vkt V)$ by the value $s(\Vkt V):=d(\Vkt L,\Vkt L')=u(\Vkt L')/E$, which we call {\it local snappability}. 
As in general a framework has several undeformed realizations $G(\Vkt V_1),\ldots,$ $G(\Vkt V_k)$ we can define a {\it global snappability} by 
$s(\Vkt L):=\min\left\{s(\Vkt V_1),\ldots, s(\Vkt V_k)\right\}$.

\subsection{Computation of the local snappability}\label{comp:snap}

We compute the critical points of the Lagrange function of Eq.\ (\ref{extendlagrange}) by using 
the homotopy continuation method (e.g.\ Bertini; cf.\ \cite{bates})
as other approaches (e.g.\ Gr\"obner base, resultant based elimination) are not promising due to the number of unknowns and degree of equations. 
By choosing a suitable reference frame we can reduce the number of unknowns by 6 for $d=3$ and  by 3 for $d=2$, respectively, which also eliminates isometries (cf.\ footnote 3). 

\begin{rmk}\label{rmk:comp}
The computation of the critical points, which depends heavily on the number of unknowns, can be a time consuming task in the first run of the homotopy. 
But if one changes the inner metric of the framework within the design process the critical points of the resulting new system of equations 
can be computed from the already known critical points of the initial system more efficiently by means of parameter homotopy \cite{bates}. \hfill $\diamond$
\end{rmk}

First of all we can restrict to the obtained real critical points as only these correspond to realizations. 
This resulting set $\mathcal{R}$ of realizations is split into a set $\mathcal{E}$, 
whose elements correspond to local extrema of $u(\Vkt V')$, and its absolute complement $\mathcal{S}=\mathcal{R}\setminus\mathcal{E}$ 
of so-called {\it saddle realizations}. 
This separation can be done by the {\it second partial derivative test} based on the  
{\it Hessian matrix} of the function $u(\Vkt V')$ or in the case of side conditions one has to use the 
bordered Hessian \cite{spring} of the Lagrangian $F(\Vkt V',\lambda)$. 
Let us denote the set of stable realizations by $\mathcal{M}\subset\mathcal{E}$, which correspond to local minima.

One way for computing the local snappability is to start at saddle points and apply gradient descent algorithms to find neighboring local minima. 
This was done in the following example of a full quad.

\begin{example}
We consider a full quad with vertices $A,B,C,D$. Its intrinsic metric is given by:
\begin{equation}
\overline{AB}=6,\quad
\overline{CD}=\sqrt{8},\quad
\overline{AC}=\overline{BD}=\sqrt{17},\quad 
\overline{AD}=\overline{BC}=\sqrt{5}.
\end{equation}
For the computation we coordinatize the vertices as follows:
\begin{equation}
A=(-x_1,0),\quad 
B=(x_1,0), \quad
C=(x_2,y_2), \quad
D=(x_3,y_3).
\end{equation}
In this way the bar from $A$ to $B$ is attached to the $x$-axis of the reference frame. 
Due to the fact that a bar cannot have zero length, we restrict ourselves to realizations of $\mathcal{M}$ and $\mathcal{S}$ where no points coincide. 
Moreover, we can assume without loss of generality that $x_1>0$ holds, as a continuous deformation between a 
realization with $x_1<0$ and a realization with $x_1>0$ has to pass $x_1=0$ ($\Leftrightarrow$ $A=B$). 
A in-depth analysis of the remaining critical points results in Fig.\ \ref{fig:fullquad}, which shows a directed graph relating
the local minima and saddle points, where the orientation points towards local minima. 
The stable realizations are illustrated in Fig.\ \ref{fig:fullquad1234} and the saddle realizations in 
Figs. \ref{fig:fullquad5678} and \ref{fig:fullquadrest}, respectively. 
Representative snaps (transitions) between stable realizations are illustrated in Fig.\ \ref{fig:quadsanp}. 

Even though this framework is globally rigid we get two undeformed realizations $G(\Vkt V_1)$ and  $G(\Vkt V_2)$, which are mirror-symmetric with respect to the 
$x$-axis (Fig.\ \ref{fig:fullquad1234}). Moreover, it is possible to snap from $G(\Vkt V_1)$ into $G(\Vkt V_2)$ (cf.\ Figs.\ \ref{fig:fullquad}	and \ref{fig:quadsanp}). 
Thus the property of being globally rigid cannot save the framework from the snapping phenomenon. 

The values for $u(\Vkt V_i)/E$ as well as $x_1,x_2,y_2,x_3,y_3$ of the stable/saddle realizations 
are given in Table \ref{table:quad}. From these values and the graph given in Fig.\  \ref{fig:fullquad}
one sees that the two local snappabilities are equal; thus $s(\Vkt L)=s(\Vkt V_{1})=s(\Vkt V_{2})=0.017411595327$ holds. 
\end{example}

\begin{figure}[t]
\begin{center} 
\begin{overpic}
    [width=90mm]{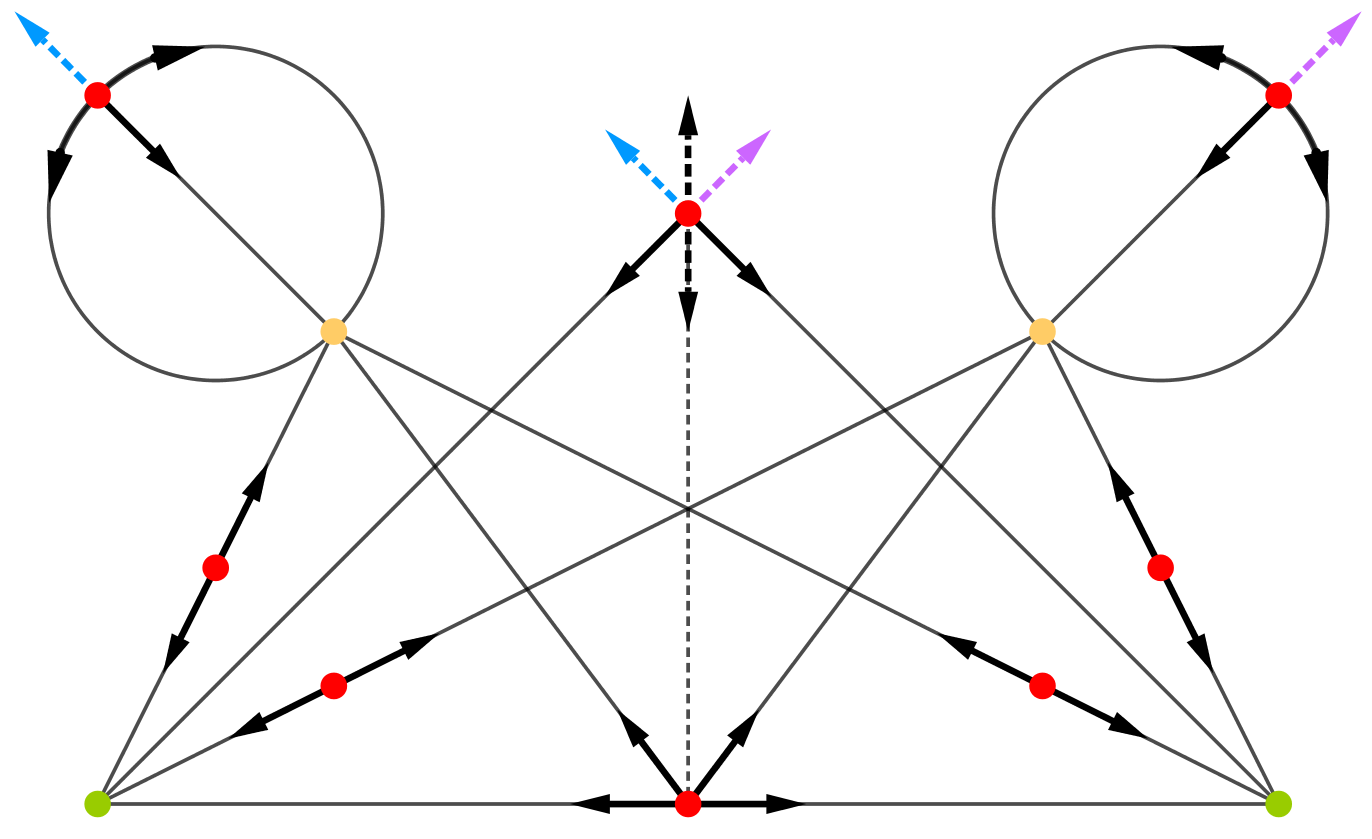} 
\begin{scriptsize}
\put(-3.5,53){$G(\Vkt V_{11})$}
\put(4,59){$_1$}
\put(9.3,49){$_1$}
\put(1.5,48){$_2$}
\put(11.5,59){$_2$}
\put(26,37){$G(\Vkt V_{3})$}
\put(95,53){$G(\Vkt V_{10})$}
\put(66,37){$G(\Vkt V_{4})$}
\put(7,19){$G(\Vkt V_{5})$}
\put(86,19){$G(\Vkt V_{6})$}
\put(25,9){$G(\Vkt V_{7})$}
\put(68,9){$G(\Vkt V_{8})$}
\put(2,-1){$G(\Vkt V_{1})$}
\put(92,-1){$G(\Vkt V_{2})$}
\put(48,-1){$G(\Vkt V_{9})$}
\put(52,45){$G(\Vkt V_{12})$}
\put(95,59){$_1$}
\put(90,49){$_1$}
\put(97.7,48){$_2$}
\put(87,59){$_2$}
\put(42.5,4){$_1$}
\put(57,4){$_1$}
\put(44,7){$_2$}
\put(55,7){$_2$}
\put(43.3,50){$_1$}
\put(43.3,41){$_3$}
\put(56,50){$_1$}
\put(56,41){$_3$}
\put(51,38){$_2$}
\put(51,53){$_2$}

\end{scriptsize} 		
  \end{overpic} 
\end{center} 
\hfill
\begin{scriptsize}
\end{scriptsize}
\caption{Directed graph relating the stable realizations (undeformed green and deformed yellow) and saddle relations (red), where the orientation points towards local minima. 
The number $i$ beside an arrow refers to the gradient flow in direction of the $i$-th smallest main curvature. 
If this number is missing, then there is only one negative main curvature direction. 
The two blue/violet dotted arrows indicate that these two flows end up in a realization obtained from $G(\Vkt V_{3/4})$ 
by reflecting it at the $y$-axis ($\Rightarrow$ $x_1<0$ which contradicts our assumption). 
The edge and arrow  between $G(\Vkt V_{12})$ and $G(\Vkt V_{9})$ are dotted as under this gradient flow all points move along the $x$-axis, 
which imply that some edge lengths become zero during the deformation. But taking small perturbations into account these zeros can be avoided. 
As this arrow appears between two saddle realizations, it means that $G(\Vkt V_{12})$ can also be deformed into  $G(\Vkt V_3)$  and $G(\Vkt V_4)$, respectively. 
The second black dotted arrow pointing upwards indicate that this flow ends up in the mirrored realization of $G(\Vkt V_9)$ ($\Rightarrow$ $x_1<0$ which contradicts our assumption). 
Finally it should be noted that the transition between the undeformed realizations $G(\Vkt V_1)$ and $G(\Vkt V_2)$ can not only be done by snaps over the 
shaky realizations $G(\Vkt V_9)$ and $G(\Vkt V_{12})$, respectively, 
but also by two subsequent snaps  from $G(\Vkt V_1)$ over $G(\Vkt V_{5/7})$ to $G(\Vkt V_{3/4})$ and further over $G(\Vkt V_{8/6})$ to $G(\Vkt V_2)$, without passing a shaky realization.
}\label{fig:fullquad}	
\end{figure}

\begin{figure}[h!]
\begin{center} 
\begin{overpic}
    [width=30mm]{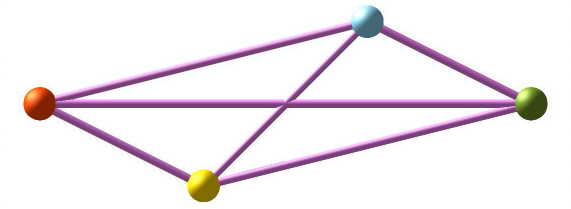} 
  \end{overpic} 
\quad
\begin{overpic}
    [width=30mm]{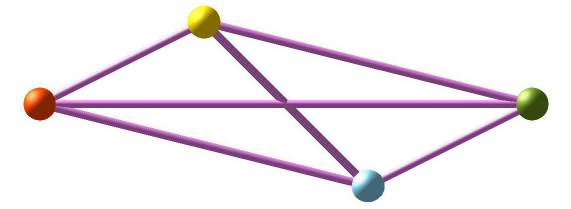} 
  \end{overpic} 
\quad
\begin{overpic}
    [width=30mm]{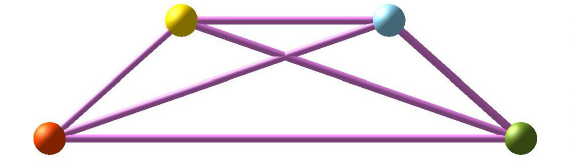} 
  \end{overpic} 
\quad
\begin{overpic}
    [width=30mm]{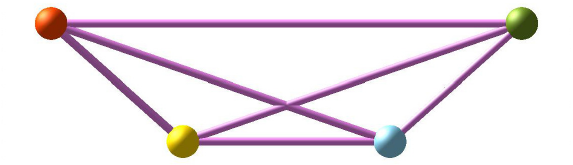} 
  \end{overpic} 	
\end{center} 
\hfill
\begin{scriptsize}
\end{scriptsize}
\caption{Local minima: The two undeformed realizations $G(\Vkt V_1)$ and  $G(\Vkt V_2)$ and the two deformed stable realizations $G(\Vkt V_3)$ and  $G(\Vkt V_4)$
(from left to right).  The point $A$ is colored red, $B$ green, $C$ blue and $D$ yellow, respectively. This color-coding is also used in the 
Figs.\ \ref{fig:fullquad5678}--\ref{fig:quadsanp}.
}\label{fig:fullquad1234}
\end{figure}

\begin{figure}[h!]
\begin{center} 
\begin{overpic}
    [width=30mm]{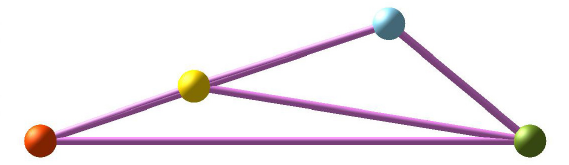} 
  \end{overpic} 
\quad
\begin{overpic}
    [width=30mm]{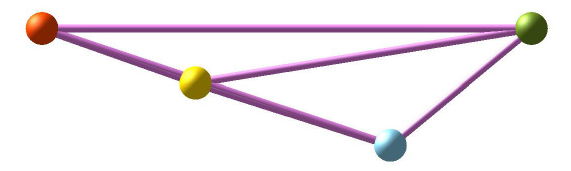} 
  \end{overpic} 
\quad
\begin{overpic}
    [width=30mm]{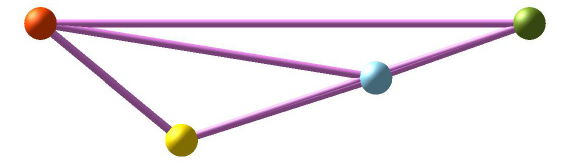} 
  \end{overpic} 
\quad
\begin{overpic}
    [width=30mm]{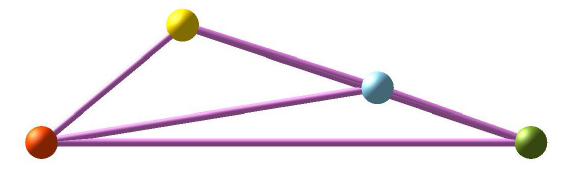} 
  \end{overpic} 	
\end{center} 
\hfill
\begin{scriptsize}
\end{scriptsize}
\caption{Saddle points: These four realizations $G(\Vkt V_5),\ldots ,G(\Vkt V_8)$ (from left to right) correspond to saddle points. 
Three points look to be collinear but they are not (cf.\ values given in Table \ref{table:quad}). 
}\label{fig:fullquad5678}	
\end{figure}

\begin{figure}[h!]
\begin{center} 
\begin{overpic}
    [width=30mm]{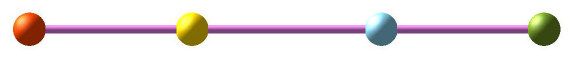} 
  \end{overpic} 
\quad
\begin{overpic}
    [width=30mm]{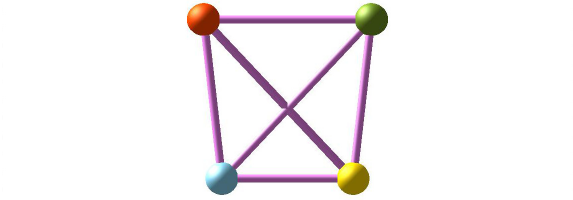} 
  \end{overpic} 
\quad
\begin{overpic}
    [width=30mm]{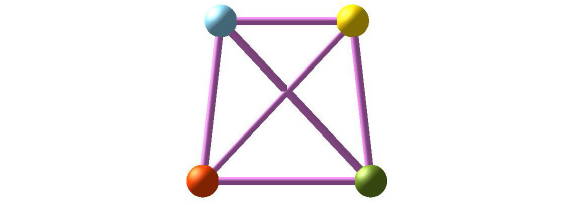} 
  \end{overpic} 
\quad
\begin{overpic}
    [width=30mm]{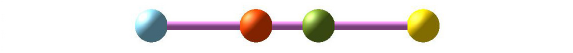} 
  \end{overpic} 	
\end{center} 
\hfill
\begin{scriptsize}
\end{scriptsize}
\caption{Saddle points: These four realizations $G(\Vkt V_9),\ldots ,G(\Vkt V_{12})$ (from left to right) correspond to saddle points. 
Moreover, the realizations $G(\Vkt V_9)$ and $G(\Vkt V_{12})$ are shaky realizations. 
}\label{fig:fullquadrest}	
\end{figure}

\begin{figure}[h!]
\begin{center}
\begin{minipage}{60mm}
\begin{center} 
\begin{overpic}
    [width=60mm]{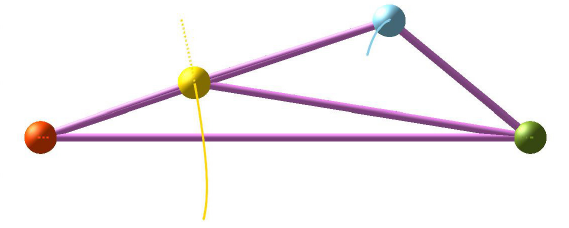} 
  \end{overpic} 
\quad \phm \\
\begin{overpic}
    [width=60mm]{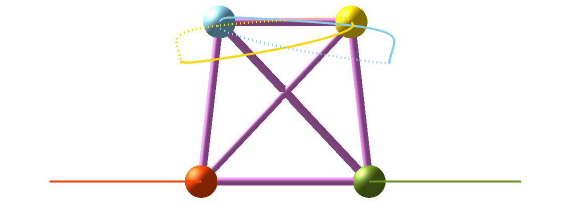} 
  \end{overpic} 
\end{center}	
\end{minipage}
\qquad
\begin{minipage}{60mm}
\begin{center} 
\begin{overpic}
    [width=60mm]{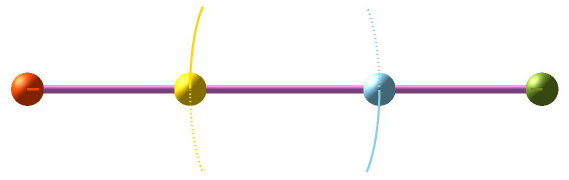} 
  \end{overpic} 
\quad
\begin{overpic}
    [width=60mm]{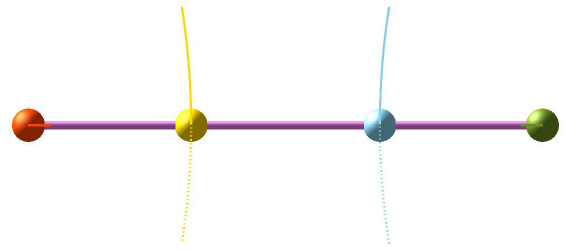} 
  \end{overpic} 	
\quad
\begin{overpic}
    [width=60mm]{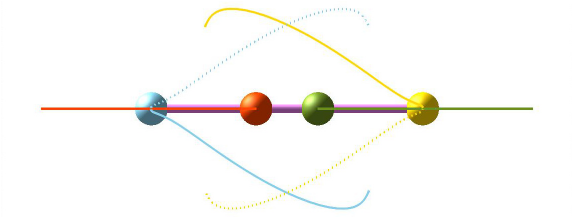} 
  \end{overpic} 		
\end{center} 
\end{minipage}
\end{center}
\hfill
\begin{scriptsize}
\end{scriptsize}
\caption{The gradient flows starting at some saddle realizations in direction of negative main curvature. 
In one direction the paths of the points are plotted solid and in the opposite direction dotted. 
Upper left: Saddle realization $G(\Vkt V_5)$ and the corresponding paths of the points towards the realizations 
 $G(\Vkt V_1)$ and $G(\Vkt V_3)$. Lower left: Following the gradient flow starting at the 
saddle realization $G(\Vkt V_{11})$ in direction of the second-smallest main curvature shows that both paths end up in 
the realization $G(\Vkt V_3)$. Therefore this realization can snap into itself over the saddle realization $G(\Vkt V_{11})$. 
Upper (center) right: Saddle realization $G(\Vkt V_9)$ and the gradient flow in direction of the smallest (second-smallest) 
main curvature. This shows the snapping between $G(\Vkt V_1)$ and $G(\Vkt V_2)$ ($G(\Vkt V_3)$ and $G(\Vkt V_4)$) over 
the shaky  realization $G(\Vkt V_9)$. By combining the different gradient flows $G(\Vkt V_1)$ can also snap into $G(\Vkt V_3)$ over 
$G(\Vkt V_9)$, but one needs more deformation energy compared to the snap over the saddle realization $G(\Vkt V_5)$ illustrated in the upper left corner 
(cf.\ Table \ref{table:quad}). 
Lower right: Saddle realization $G(\Vkt V_{12})$ and the gradient flow in direction of the third-smallest main curvature. This shows an 
alternative snapping between $G(\Vkt V_1)$ and $G(\Vkt V_2)$, which also needs more deformation energy as the one over  $G(\Vkt V_{12})$ 
illustrated in the upper right corner (cf.\ Table \ref{table:quad}). 
Finally note that the three snaps illustrated in the right column also demonstrate Theorem \ref{thm122} 
and on the infinitesimal level Corollary \ref{cor:dif}.
}\label{fig:quadsanp}	
\end{figure}

\begin{table}[!]
\begin{center}
\begin{footnotesize}
\begin{tabular}[h]{c|cccccc}
 & $x_1$ &  $x_2$ & $y_2$  & $x_3$   & $y_3$ & $u(\Vkt V_i)/E$ \\ \hline\hline
 $\Vkt V_1$ & 3 & 1  & 1 & -1   & -1  & 0 \\ \hline 
 $\Vkt V_2$ & 3 & 1  & -1 & -1  &  1  & 0  \\ \hline 
 $\Vkt V_3$ & 2.873803815106 & 1.264482434547  &  1.435718952668 & -1.264482434547  &   1.435718952668 	& 0.014510412969  \\ \hline 
 $\Vkt V_4$ & 2.873803815106 & 1.264482434547  & -1.435718952668 & -1.264482434547  &  -1.435718952668 	& 0.014510412969   \\ \hline 
 $\Vkt V_5$ & 2.984971064849 & 1.262815500919  &  1.420716567414 & -1.110499467187  &   0.662434022361 	& 0.017411595327  \\ \hline 
 $\Vkt V_6$ & 2.984971064849 & 1.262815500919  & -1.420716567414 & -1.110499467187  &  -0.662434022361 	& 0.017411595327   \\ \hline 
 $\Vkt V_7$ & 2.984971064849 & 1.110499467187  & -0.662434022361 & -1.262815500919  &  -1.420716567414 	& 0.017411595327   \\ \hline 
 $\Vkt V_8$ & 2.984971064849 & 1.110499467187  &  0.662434022361 & -1.262815500919  &   1.420716567414 	& 0.017411595327   \\ \hline 
 $\Vkt V_9$ & 3.139661485127 & 1.152372944387  & 0               & -1.152372944387  &  0  								& 0.029190294037  \\ \hline 
 $\Vkt V_{10}$ & 1.023109368578 & -0.801485412727  & -1.942382880068 &  0.801485412727  & -1.942382880068 & 0.544598068230     \\ \hline 
 $\Vkt V_{11}$ & 1.023109368578 & -0.801485412727  &  1.942382880068 &  0.801485412727  &  1.942382880068 & 0.544598068230   \\ \hline 
 $\Vkt V_{12}$ & 0.377904764722 & -1.656968246848  & 0               &  1.656968246848 &  0 							& 0.606663470611   \\ \hline 
\end{tabular}
\end{footnotesize}
\end{center}
\caption{Coordinates of stable/saddle realizations and their corresponding $u(\Vkt V_i)/E$ value.} \label{table:quad}
\end{table}

In practice one starts with the saddle realization $G(\Vkt V_i)$ with the lowest value for $u(\Vkt V_i)/E$ 
and compute gradient flows towards the stable realizations, hoping that one ends up in the 
undeformed realization under consideration. If one does not find such a descent path, then one repeats the procedure for the 
saddle realization with the next higher value for $u(\Vkt V_i)/E$. 
The problem of this approach is that one has no guarantee to 
detected all local minima which can be reached from a saddle point by a descent path. 
With guarantee one can only give a lower bound\footnote{Note that $o(\Vkt L)$ can be seen as a kind of separation bound \cite{herman} of the real roots of the polynomial $u(\Vkt V')$ in 
terms of the intrinsic metric of framework inducing this polynomial.} for the global snappability by the following value:

\begin{lem}\label{minorant}
Let us assume that  $G(\Vkt V^-)\in\mathcal{S}$ implies the minimal value for $d(\Vkt L,\Vkt L^-)$ for all elements of $\mathcal{S}$. 
As a result $o(\Vkt L):=d(\Vkt L,\Vkt L^-)\leq s(\Vkt L)$ has to hold. 
\hfill $\diamond$
\end{lem}

Therefore we want to present a more sophisticated approach, which allows us to check directly if 
a saddle realization and a stable realization can be deformed continuously into each other, 
whereby the deformation energy density has to decrease monotonically. 
The minor drawback of this method is that it only works for 
pin-jointed body-bar framework, which are isostatic.

\subsection{Pin-jointed body-bar frameworks with minimal rigidity under affine deformations}\label{isostatic}  

We assume that the given pin-jointed body-bar framework $G(\mathcal{K})$ is isostatic\footnote{Note that the isostaticity of a spatial framework is 
a problem for its own and not treated within this article (see e.g.\ \cite{cfgsw,gfs})}. 
By replacing the bodies by the corresponding globally rigid subframeworks, the equivalent body-bar framework $G_*(\mathcal{K}_*)$ is also isostatic, 
if all polyhedra (polygonal panels) are tetrahedral (triangular). But in the general case the subframeworks are overbraced. 

The edges of the overbraced subframeworks are involved in the computation of the density function given in Lemma \ref{basic}. 
One can get rid of this property by allowing only affine (homogeneous) deformations of the bodies. 
In this way every polyhedron (polygonal panel) can be represented by a tetrahedron $V_a,V_b,V_c,V_d$ (triangle $V_a,V_b,V_c$) and the 
remaining vertices of the deformed polyhedron (polygonal panel) can then be obtained by the affine transformation  determined by 
$V_i\mapsto V'_i$ for $i\in\left\{a,b,c,d\right\}$ (resp.\ $i\in\left\{a,b,c\right\}$).

As a consequence we can consider a minimal set of lengths $L_{ij}$ which contains the lengths of the bars $\in\mathcal{G}$ plus for 
each polyhedron (polygonal panel) we get six (three) additional lengths. We collect these lengths within the vectors $\kVkt L:=(\ldots, L_{ij},\ldots)^T\in \RR^a$  
and $\kVkt Q:=(1, \ldots ,Q_{ij}, \ldots)^T\in \RR^{a+1}$ with $a\leq b$.
Using this notation the density function can be rewritten as follows:
\begin{equation}\label{eq:u_affin}
u(\kVkt L'):=\frac{\sum_{ab\in\mathcal{G}} U_{ab}(\kVkt L')+
\sum_{i=1}^p \Vol(B_i^2)\left[\frac{U_{abc}(\kVkt L')}{ \Vol_{abc}}  \right] + 
\sum_{j=p+1}^q \Vol(B_j^3)\left[\frac{ U_{abcd}(\kVkt L')}{ \Vol_{abcd}} \right]}
{\sum_{ab\in\mathcal{G}} \Vol_{ab}+
\sum_{i=1}^p \Vol(B_i^2) + 
\sum_{j=p+1}^q \Vol(B_j^3)}
\end{equation}
for an arbitrary $abc\in\mathcal{C}_i$ and $abcd\in\mathcal{C}_j$, respectively. 
As we can replace $L'_{ij}$ in $u(\kVkt L')$ by $\|\Vkt v'_i-\Vkt v'_j\|$ the function $u$ can be computed in dependence of $\kVkt V'\in\RR^{\tilde wd}$ with
$\tilde w=r+3p+4(q-p)$, where $\kVkt V'$ contains the vectors of the vertices $V_1,\ldots, V_r$ as well as three vertices of each 
polygonal panel and four vertices of each polyhedron, respectively.
With respect to the resulting density function $u(\kVkt V')$ we compute the set $\mathcal{M}$ of stable realizations and the set $\mathcal{S}$ of saddle realizations. 
Moreover, the following theorem holds true:

\begin{thm}\label{thm:neu}
If an isostatic pin-jointed body-bar framework snaps out of a stable realization $G(\Vkt V)$ 
by applying the minimum strain energy needed to it, then the framework passes
a realization $G(\Vkt V')$ at the maximum state of deformation (under the assumption that each body is deformed affinely), 
which is either shaky or contains at least a body of reduced dimension.
\end{thm}

\noindent
Proof: 
For the analysis of the equations  $\nabla_{\hspace{-0.5mm}a}\, u(\kVkt V')=0$ we use the 
 sum rule for derivatives to study each summand separately. 

Let us assume that the vertex $V_a\in B_i$ is a vertex of a tetrahedron (triangle) representing a polyhedron (polygonal panel). 
From Theorem \ref{thm:critic} it is already known that $\nabla_{\hspace{-0.5mm}a}\, U_{abcd}(\Vkt v_a',\Vkt v_b',\Vkt v_c',\Vkt v_d')$ 
(resp.\ $\nabla_{\hspace{-0.5mm}a}\, U_{abc}(\Vkt v_a',\Vkt v_b',\Vkt v_c')$)
can be written as a linear combination of the involved vectors (cf.\ Eq.\ (\ref{test1}) and Eq.\ (\ref{test2}), respectively). 
Therefore we are only left with the partial derivative of the elastic strain energy of a bar $ex\in\mathcal G$ with $V_e\in B_i$, i.e.\ 
\begin{equation} \label{us:lincomb}
\Vkt v_e=\Vkt v_a+\xi(\Vkt v_b-\Vkt v_a) + \upsilon(\Vkt v_c-\Vkt v_a) + \zeta (\Vkt v_d-\Vkt v_a) 
\end{equation}
where $\zeta=0$ holds for a polygonal panel. We get:
\begin{equation}
\nabla_{\hspace{-0.5mm}a}\, U_{ex}(\Vkt v_a',\Vkt v_b',\Vkt v_c',\Vkt v_d',\Vkt v_x') = \tfrac{\Area_{ex}(L_{ex}'^2-L_{ex}^2)}{2L_{ex}^3}
\left[\xi(\Vkt v'_b-\Vkt v'_a) + \upsilon(\Vkt v'_c-\Vkt v'_a) + \zeta (\Vkt v'_d-\Vkt v'_a)+\Vkt v'_a-\Vkt v'_x\right](1-\xi- \upsilon-\zeta).
\end{equation}
This shows that $\nabla_{\hspace{-0.5mm}a}\, u(\kVkt V')$ can be written as a linear combination of the partial derivatives of the squared distances 
of vertices linked by tetrahedral (triangular) edges or green edges (edges $\in\mathcal{G}$). 

As a consequence, each solution of the $\tilde w$ equations of the form $\nabla_{\hspace{-0.5mm}i}\, u(\kVkt V')=0$ 
can be associated with a ${\tilde wd}$-dimensional stress-vector ${\omega}$. 
If this vector differs from the zero vector then the solution has to imply  a rank defect of the 
square matrix, whose columns are the gradient vectors of the realization equations $c_1,\ldots ,c_{\tilde wd}$. 
This rank defect either corresponds with a shaky configuration of the framework $G(\Vkt V')$ or arises 
from the shakiness of a substructure (i.e.\ a tetrahedron or a triangle) substituting a body. A tetrahedron (triangle) 
can only become infinitesimal flexible if its dimension is reduced to at least a plane  (line). 
\bigskip \hfill $\BewEnde$

An advantage of Theorem \ref{thm:neu} over Theorem \ref{thm121} is that the property of the realization $G(\Vkt V')$ concerns the pin-jointed body-bar framework 
and not the equivalent bar-joint framework (assumed that the given pin-jointed body-bar framework is not a bar-joint framework). \bigskip

\noindent
{\bf Is the assumption of affine deformations really a restriction?}
For our preferred choice of $\nu=1/2$ the volume has to be constant. This has the following 
consequences for the deformation of polyhedra and polygonal panels, respectively.

\begin{enumerate}[$\bullet$]
\item
Assume a polyhedron $B_j^3(n_j)$ 
with $n_j>4$ is given. Moreover, we assume that $(b,c,d,e)$, $(a,c,d,e)$, $(a,b,d,e)$ and $(a,b,c,e)$ belong to 
the index set $\mathcal{C}_j$ of non-degenerated tetrahedra. 
Therefore one can compute the barycentric coordinates of $V_e$ with respect to the tetrahedron $V_a,\ldots ,V_d$, which  
equal the ratio of the oriented volumes $(\Vol_{bcde}:\Vol_{acde}:\Vol_{abde}:\Vol_{abce})$ 
according to \cite{moebius}.

Now we assume that the tetrahedron $V_a,\ldots ,V_d$ was deformed isochoricly into 
$V'_1,\ldots ,V'_4$. Therefore not only the volume remains constant under the deformation, but 
also its orientation\footnote{A change in orientation can only happen in a flat pose of the 
tetrahedron, which has zero volume.}. As a consequence the mapping from $V_1,\ldots ,V_4$ to 
$V'_1,\ldots ,V'_4$ can be written as
\begin{equation}\label{affin:map}
\Vkt v'_i=\Vkt A\Vkt v_i + \Vkt a \quad \text{with} \quad \det(\Vkt A)=1 \quad \text{for} \quad i=a,b,c,d.
\end{equation}
Then $V'_e$ is uniquely determined by its barycentric coordinates 
$(\Vol'_{bcde}:\Vol'_{acde}:\Vol'_{abde}:\Vol'_{abce})$
with respect to the tetrahedron $V'_a,\ldots ,V'_d$ which have to equal 
the above given homogenous 4-tuple. This already implies that $V'_e$ and 
$V_e$ are also in the affine correspondence of  Eq.\ (\ref{affin:map}) 
(e.g.\ \cite[page 61]{aichholzer}). 
\item
Assume a polygonal panel $B_i^2(n_i)$ 
with $n_i>3$ is given. Moreover, we assume that $(b,c,d)$, $(a,c,d)$ and $(a,b,d)$ belong to 
the index set $\mathcal{C}_i$ of non-degenerated triangles. 
According to Section \ref{sec:pratio} the panel height $h'_{abc}$ of the deformed triangular subpanel can 
be computed as $\Vol_{abc}/\Area'_{abc}$. According to the second item of Remark \ref{rmk:lemma1} this height 
can be identified with the height  $h'_i$ of the deformed panel. 
In order that  the heights $h'_{bcd}$, $h'_{acd}$ and $h'_{abd}$ 
of the other three triangular subpanels are equal to $h'_i$ the following ratio has to hold:
\begin{equation}
\Area'_{abc}:\Area'_{bcd}:\Area'_{acd}:\Area'_{abd} = 
\Area_{abc}:\Area_{bcd}:\Area_{acd}:\Area_{abd}.
\end{equation}
By means of planar barycentric coordinates it can be seen that all four points have to be mapped 
by an affine transformation:
\begin{equation}
\Vkt v'_i=\Vkt A\Vkt v_i + \Vkt a  \quad \text{for} \quad i=a,b,c,d.
\end{equation}
Note that in this case the matrix $\Vkt A$ is not restricted to $\det(\Vkt A)=1$ but one only has to assume that 
$\Area'_{abc}\neq 0$. For $d=2$ the condition $\Area'_{abc}\neq 0$ is equivalent with $\det(\Vkt A)\neq 0$.
\end{enumerate}
These considerations show that the assumption of affine deformations is no restriction in the case 
where the inner graph of the polyhedron (polygonal panel) is a complete graph on at most four (three) vertices. 
In the more general case of global rigidity, which is only known for $\RR^2$ (and still open for $\RR^3$; 
cf.\ Section \ref{sec:outline}) and characterized by 3-connectivity and redundant rigidity 
of the graph \cite{jackson}, this assumption is maybe\footnote{To clarify this open problem, one has to study in 
more detail the properties of globally rigid inner graphs (cf.\ Definition \ref{def:ingr}).} restrictive. 
Taking these possible minor restrictions into account, we can compute the local snappability 
in an efficient way as follows. \bigskip

\noindent
{\bf Algorithm for computing the local snappability.}
Given is an undeformed realization $G(\kVkt V)\in\mathcal{M}$ and we want to determine $s(\kVkt V)$. To do so, we consider the 
saddle realization $G(\kVkt V')\in\mathcal{S}$ which has the minimal value for $d(\kVkt L,\kVkt L')$ and define the transformation 
\begin{equation}\label{imply}
\kVkt Q_t:=\kVkt Q+t(\kVkt Q'-\kVkt Q) \quad \text{with}\quad t\in[0,1].
\end{equation}
This gradient flow in the space of squared leg lengths (cf.\ Remark \ref{rmk:critic}) 
implies a path $\kVkt L_t$ in $\RR^a$ between $\kVkt L$ and $\kVkt L'$. 
Along this path  
the deformation energy of each tetrahedron $U_{abcd}$, triangular panel  $U_{abc}$ as well as bar $U_{ab}$ 
is {\it monotonic increasing} with respect to the path parameter $t$. 
This ensures that only the minimum mechanical work needed is applied on the framework to reach $G(\kVkt V')$.  
This results from Lemma \ref{basic}, as $U_{abcd}(\kVkt L_t)$, $U_{abc}(\kVkt L_t)$ as well as $U_{ab}(\kVkt L_t)$  are quadratic functions in $t$, 
which are at their minima for $t=0$. 
The path $\kVkt L_t$ corresponds to different 1-parametric deformations  of realizations in $\RR^d$. 
If among these a deformation $G(\kVkt V_t)$ with the property 
\begin{equation}\label{property}
G(\kVkt V_t)\big|_{t = 0}=G(\kVkt V),\quad G(\kVkt V_t)\big|_{t = 1}=G(\kVkt V') 
\end{equation}
exists, then the given realization $G(\kVkt V)$ is deformed into $G(\kVkt V')$ under $\kVkt L_t$. 
Computationally the property (\ref{property}) can easily be checked as follows: 
We consider the set of algebraic {\it realization equations} $c_1,\ldots ,c_n$ implied by the framework (cf.\ Section \ref{sec:outline}). 
Due to Eq.\ (\ref{imply}) the equations, which correspond to bar constraints, depend linearly on $t$. 
We have to track the path of the solution $\kVkt V$ of this algebraic system while $t$ is increasing from zero to one. 
This is a homotopy continuation problem which can be solved efficiently e.g.\ by the software Bertini \cite[Section 2.3]{bates}. 
 
\begin{rmk}
Note that this approach has to be adapted in the special case that the undeformed realization $G(\kVkt V)\in\mathcal{M}$ is shaky, 
as $\kVkt V$ is a singular solution of the algebraic system for $t=0$. 
In this case one has to solve the set of algebraic equations $c_1,\ldots ,c_n$ for a random value $t_*\in(0,1)$. 
The resulting solutions are then tracked by homotopy continuation back to the value $t=0$. 
At least two paths have to lead down to $\kVkt V$. The corresponding solutions at  $t=t_*$ of these paths   
are then tracked by homotopy continuation up to the value $t=1$ to check if one of them ends up at $\kVkt V'$.
\hfill $\diamond$
\end{rmk}

If no deformation with the property (\ref{property}) exists then we redefine 
$\mathcal{S}$ as $\mathcal{S}\setminus\left\{ G(\kVkt V')\right\}$ and run again the procedure explained in this paragraph until we 
get the sought-after realization implying $s(\kVkt V)$. 
If we end up with $\mathcal{S}=\left\{\,\right\}$ then we set $s(\kVkt V)=\infty$.

\begin{rmk}\label{rem:notconstvol}
Even if  $G(\kVkt V)$ and $G(\kVkt V')$ have the same volume, the deformation implied by Eq.\ (\ref{imply}) is in general not isochoric. \hfill $\diamond$
\end{rmk}


\section{Local and global singularity-distance} \label{sec:dist}

In the following we want to determine the real point $\Vkt V'''$ of the {\it shakiness variety} $V(J)$ (cf.\ end of Section \ref{sec:outline}) minimizing the 
value $d(\Vkt L,\Vkt L''')$,  where $G(\Vkt V)$ is the given undeformed realization of a  pin-jointed body-bar framework $G(\mathcal{K})$. 
In addition there should again exist a 1-parametric deformation of $G(\Vkt V)$ into $G(\Vkt V''')$ such that the deformation energy density has to 
increase monotonically. If this is the case we call $\varsigma(\Vkt V)=d(\Vkt L,\Vkt L''')=u(\Vkt L''')/E$ the {\it local singularity-distance}. 
By taking the minimum of all local singularity-distances of possible undeformed realizations of a framework we get the {\it global singularity-distance} $\varsigma(\Vkt L)$; i.e.\ 
$\varsigma(\Vkt L):=\min\left\{\varsigma(\Vkt V_1),\ldots, \varsigma(\Vkt V_k)\right\}$.

In the general case one has to compute the local minima of the Lagrangian 
\begin{equation}\label{extendlagrange2}
F(\Vkt V',\lambda)= u(\Vkt V') 
-{\lambda}_1f_1-\ldots -{\lambda}_{\phi}f_{\phi} - \lambda_{\phi+1}g_1- \ldots -\lambda_{\phi+\gamma}g_{\gamma} \quad \text{with}\quad
\lambda:=({\lambda}_1,\ldots,{\lambda}_{\phi+\gamma}), 
\end{equation}
where we recall that $g_1,\ldots ,g_{\gamma}$  are the generators of the ideal $J$ of the shakiness variety $A(J)$ and  $f_1,\ldots,f_{\phi}$
denote the side conditions for an isochoric deformation if this is desired. 
Starting from the corresponding realizations one can apply again gradient descent algorithms 
with respect to the function of the deformation energy density in order to find the neighboring stable realizations. 
Clearly, this strategy is faced with the same problems as already mentioned in Section \ref{comp:snap}. 
But for pin-jointed body-bar frameworks with minimal rigidity under affine transformations 
we are able to prove the following statements  (Theorem \ref{thm:ident} and Corollaries \ref{cor2} and \ref{cor:def}).

\begin{thm}\label{thm:ident}
If the undeformed realization $G(\Vkt V)$ is not shaky, then 
the local snappability $s(\Vkt V)$ is a lower bound on the local singularity-distance $\varsigma(\Vkt V)$; i.e.\ $s(\Vkt V) \leq \varsigma(\Vkt V)$. 
If the realization $G(\Vkt V')$,  which implies the snappability $s(\Vkt V)$, is shaky, then the equality holds.
\end{thm}

\noindent
Proof: 
We show the relation $s(\Vkt V) \leq \varsigma(\Vkt V)$ indirectly by assuming $\varsigma(\Vkt V)< s(\Vkt V)$. We denote the shaky realization implying $\varsigma(\Vkt V)$ 
by $G(\Vkt V'')$ which  corresponds to $\kVkt L''\in\RR^a$. 
In analogy to Eq.\  (\ref{imply}) we consider the relation
\begin{equation}
\kVkt Q_t:=\kVkt Q+t(\kVkt Q''-\kVkt Q) \quad \text{with}\quad t\in[0,1]
\end{equation}
defining a path $\kVkt L_t$ in $\RR^a$ between $\kVkt L$ and $\kVkt L''$, which corresponds to a set 
of 1-parametric deformations $\left\{G(\Vkt V_t^1),G(\Vkt V_t^2), \ldots \right\}$. 
A subset $\mathcal{D}$ of this set has the property $G(\Vkt V_t^i)|_{t = 1}=G(\Vkt V'')$ where $\#\mathcal{D}>1$ holds as 
$G(\Vkt V'')$ is shaky \cite{wohlhart,stachel_wunderlich}. 
Therefore the framework can snap out of $G(\Vkt V)$ over $G(\Vkt V'')$ which contradicts $\varsigma(\Vkt V)< s(\Vkt V)$ 
($\Longrightarrow$ $s(\Vkt V) \leq \varsigma(\Vkt V)$). \bigskip \hfill $\BewEnde$

In the case of Theorem \ref{thm:ident} the local snappability gives the radius of a guaranteed singularity-free sphere in the 
space of intrinsic framework metrics for a non-shaky realization. 
Note that in the space $\RR^b$ of squared edge lengths $Q_{ij}$ this singularity-free zone is 
bounded by a hyperellipsoid due to the third item of Remark \ref{rmk:lemma1}. 

Moreover, Theorem \ref{thm:ident} implies the following statement:

\begin{cor}\label{cor2}
If  none of the undeformed  realizations of a framework $G(\Vkt V_1),\ldots, G(\Vkt V_k)$ is shaky, then the 
global snappability $s(\Vkt L)$ is a lower bound on the global singularity-distance $\varsigma(\Vkt L)$.
\end{cor}

Note that in case of Corollary \ref{cor2} also $o(\Vkt L)\leq\varsigma(\Vkt L)$ has to hold due to Lemma \ref{minorant}. 
Therefore $\varsigma(\Vkt L)$ as well as $o(\Vkt L)$ are radii of guaranteed singularity-free spheres in the 
space of intrinsic framework metrics for any of the undeformed realizations. 

Moreover, we can make the following statement on the reality of deformations:

\begin{cor}\label{cor:def}
The deformation associated with the local snappability $s(\Vkt V)=d(\Vkt L,\Vkt L')$, which is implied by Eq.\ (\ref{imply}), is  
guaranteed to be real, if not both realizations  $G(\Vkt V)$ and $G(\Vkt V')$ are shaky.
\end{cor}

\noindent
Proof: 
A real solution of an algebraic set of equations can only change over into a complex one through a double root, 
which corresponds either to a (1) shaky realization or to a (2) body of reduced dimension. 
Case (1) is impossible due to Theorem \ref{thm:ident}. Moreover, case (2) can also not hold, which can be shown 
in the same way as in the proof of Theorem \ref{thm:ident}. 

Therefore the entire path has to be real if at least one of the two realizations is not shaky. \bigskip \hfill $\BewEnde$

In the following example we want to demonstrate the results obtained so far.




\begin{example} 
We consider a closed serial chain composed of four directly congruent tetrahedral chain elements, which are jointed by 
four hinges. The studied example was given by  Wunderlich \cite{wunderlich_gelenksviereck} 
and is illustrated in Fig.\ \ref{fig:4rloop1}.	It has a threefold reflexion symmetry with respect to three 
copunctal lines, which are pairwise orthogonal. Using them as axes of a Cartesian frame, the vertices can be 
coordinatized as follows:
\begin{align}
A_1&=(u_1,v_1,w_1)^T &\quad A_2&=(-u_2,v_2,w_2)^T &\quad A_3&=(-u_1,-v_1,w_1)^T &\quad A_4&=(u_2,-v_2,w_2)^T \\
B_1&=(u_1,-v_1,-w_1)^T &\quad B_2&=(u_2,v_2,-w_2)^T &\quad B_3&=(-u_1,v_1,-w_1)^T &\quad B_4&=(-u_2,-v_2,-w_2)^T 
\end{align}
The intrinsic metric of the framework is given by the following assignment:
\begin{equation}
\begin{split}
\overline{A_1A_2}&=\overline{A_2A_3}=\overline{A_3A_4}=\overline{A_4A_1}=\overline{B_1B_2}=\overline{B_2B_3}=\overline{B_3B_4}=\overline{B_4B_1}=
\tfrac{25\sqrt{3} - 15}{(3\sqrt{2} + 10)\sqrt{3} - 3\sqrt{2} + 6} \\
\overline{A_1B_2}&=\overline{A_2B_3}=\overline{A_3B_4}=\overline{A_4B_1}=
\tfrac{15 + 5\sqrt{3}}{(3\sqrt{2} + 10)\sqrt{3} - 3\sqrt{2} + 6} \\
\overline{A_1B_4}&=\overline{A_2B_1}=\overline{A_3B_2}=\overline{A_4B_3}=
\tfrac{45 - 5\sqrt{3}}{(3\sqrt{2} + 10)\sqrt{3} - 3\sqrt{2} + 6} \\
\overline{A_1B_1}&=\overline{A_2B_2}=\overline{A_3B_3}=\overline{A_4B_4}=
\tfrac{15\sqrt{2}(\sqrt{3} - 1}{(3\sqrt{3} - 3)\sqrt{2} + 10\sqrt{3} + 6}
\end{split}
\end{equation}
which has the property that the average edge length equals 1. 
We consider the two undeformed realizations $G(\Vkt V_1)$ and $G(\Vkt V_2)$ of the chain illustrated in Fig.\ \ref{fig:4rloop1}, which can be computed according to the procedure given in \cite{wunderlich_gelenksviereck}. 
The vector $(u_1,v_1,w_1,u_2,v_2,w_2)$ which corresponds to $\Vkt V_1$ is given by:
\begin{equation}
(0.802729630788,0.207761716516,0.207761716516,0.169636731183,0.655425998948,0.239902565915)
\end{equation}
The coordinates of $\Vkt V_2$ are obtained by the following exchange of the
coordinate entries of $\Vkt V_1$: 
$u_1 \leftrightarrow v_2$, $v_1\leftrightarrow u_2$ and $w_1\leftrightarrow w_2$.
Note that the framework can snap between the two realizations $G(\Vkt V_1)$ and $G(\Vkt V_2)$. 
The shaky saddle realization, which has to be passed during the snap, is denoted by $G(\Vkt V')$.

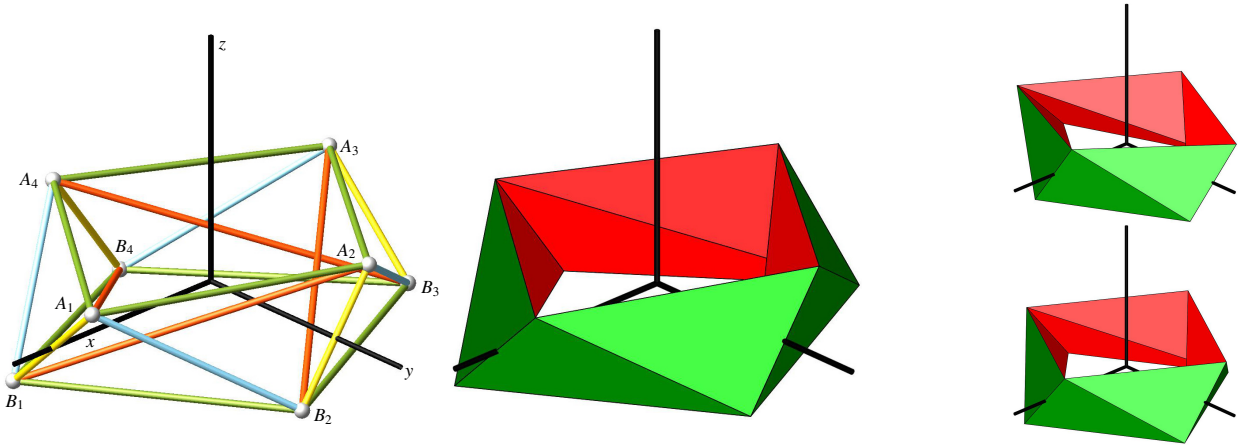
\begin{figure}[t]
\begin{center} 
\begin{minipage}{110mm}
\begin{overpic}
    [width=55mm]{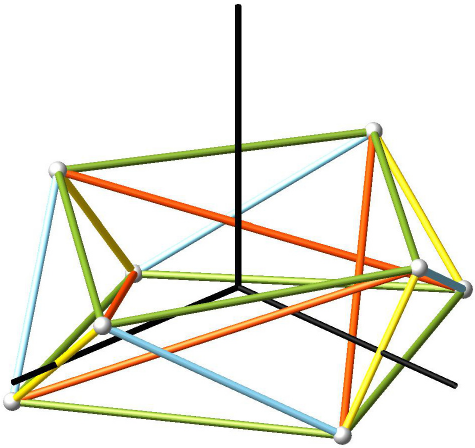}  
\begin{scriptsize}
\put(0.5,3.5){$B_1$}
\put(75,0){$B_2$}
\put(101,31){$B_3$}
\put(27,41){$B_4$}
\put(12.6,27){$A_1$}
\put(80.7,40){$A_2$}
\put(81.5,65){$A_3$}
\put(4,57){$A_4$}
\put(20,18){$x$}
\put(97,11){$y$}
\put(52,90){$z$}

\end{scriptsize} 		
  \end{overpic} 
	\quad
	\begin{overpic}
    [width=55mm]{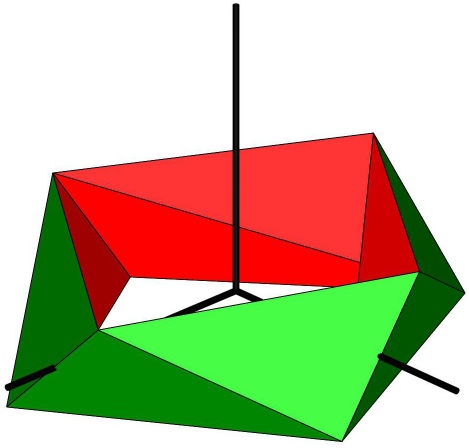} 
\begin{scriptsize}
\end{scriptsize} 		
  \end{overpic} 
\end{minipage}\hfill	
\begin{minipage}{40mm}
\begin{overpic}
    [width=30mm]{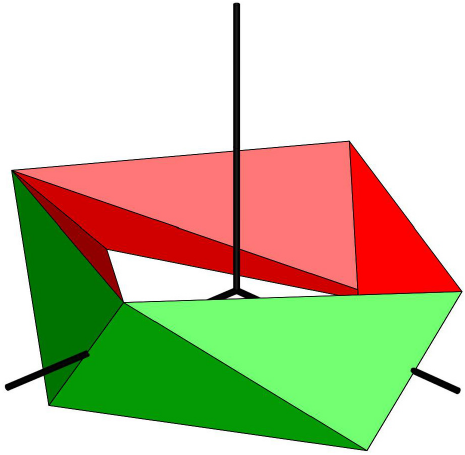} 
  \end{overpic}
\begin{overpic}
    [width=30mm]{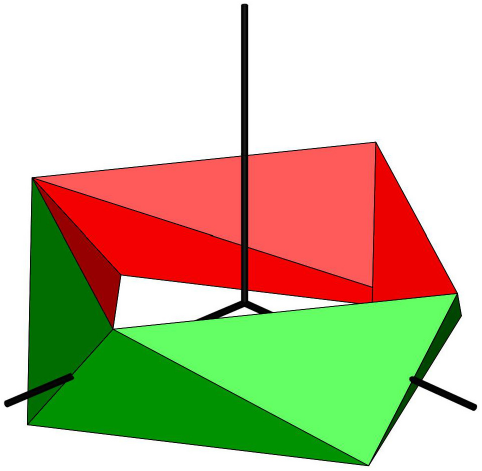} 
  \end{overpic}	
\end{minipage}	
\end{center} 
\hfill
\begin{scriptsize}
\end{scriptsize}
\caption{Left: Illustration of the realization $G(\Vkt V_1)$ as a bar-joint framework, where 
bars of equal length have the same color. The four tetrahedra are hinged along the yellow bars. 
Moreover, the coordinate frame is displayed where the axes have a length of 1. 
Center: The same configuration as on the left side but illustrated with panels instead of bars. 
Congruent triangular panels are again  same colored  (either red or green). 
Right: At the top the second realization $G(\Vkt V_2)$ is visualized and at the bottom the shaky realization $G(\Vkt V')$. 
An animation of the snapping behavior can be downloaded from \cite{data2021}.
}\label{fig:4rloop1}	
\end{figure}

In general a closed chain composed of four tetrahedra results in an overbraced bar-joint framework, as one can remove e.g.\ the bars $A_1B_4$ and $B_1B_4$ to get a minimal rigid structure. 
But under the assumed threefold symmetry resulting in directly congruent chain elements the framework is minimal rigid, 
as the input of the six edge lengths $\overline{A_1B_1}$, $\overline{A_1B_2}$, $\overline{A_1A_2}$, $\overline{A_2B_1}$, $\overline{A_2B_2}$ 
and $\overline{B_1B_2}$ already determine the six values $u_1,u_2,v_1,v_2,w_1,w_2$ coordinatizing the four involved points and therefore the complete structure.

\begin{table}
\begin{center}
\begin{footnotesize}
\begin{tabular}[h]{ll|cccc}
Chain elements & & \# tracked paths & $\# \mathcal{R}$ & $\# \mathcal{S}$     & $\varsigma(\Vkt L) =s(\Vkt L)$ \\ \hline\hline
bar-joint & 							& $729$  & $113$  & 96     &  $6.762914466510\cdot 10^{-7}$ \\ \hline
panel-hinge &  $\nu=1/2$ 	& $729$  & $161$  & 144    &  $9.363722223978\cdot 10^{-6}$\\ 
						& $\nu=1/4$ 	& $729$  & $137$  & 120    &  $1.052544771247\cdot 10^{-5}$ \\ 
						& $\nu=0$ 		& $729$  & $129$  & 112    &  $1.261816856140\cdot 10^{-5}$ \\ \hline
tetrahedra & $\nu=1/2$ 		& $279\,936$ 	& $49$	   & $33$   & $3.289330211161\cdot 10^{-5}$  \\ 
	& $\nu=1/2$ (simple) 		& $2\,187$ & $24$  & $20$  			& \ditto \\ 
						& $\nu=1/4$ 	& $729$ & $179$  & $154$   			  & $3.946472039856\cdot 10^{-5}$  \\ 
						& $\nu=0$ 		& $729$ & $178$  & $155$   				& $4.932700715589\cdot 10^{-5}$  \\ \hline
\end{tabular}
\end{footnotesize}
\end{center}
\caption{Computational data: Note that the computation of the set $\mathcal{R}$ was done by a 
total degree homotopy using Bertini. 
For the case of tetrahedral chain elements under the assumption $\nu=1/2$, the 
set $\mathcal{S}$ has to be filtered out from $\mathcal{R}$ by using a second-derivative test based
on the bordered Hessian \cite{spring} due to the isochoricity side condition.
} \label{tab0:loop}
\end{table}

We can interpret the chain elements as bar-joint frameworks, panel-hinge frameworks or as tetrahedra. 
Moreover, in the case of triangular panels and tetrahedra we compute the snappability with respect to three different Poisson ratios $\nu=0,\tfrac{1}{4},\tfrac{1}{2}$. 
The computational data for the different cases is summarized in the Tables \ref{tab0:loop} and \ref{tab2:loop}. 
Note that independent of the interpretation we get $s(\Vkt L)=s(\Vkt V_1)=s(\Vkt V_2)$.
The corresponding saddle realizations  $G(\Vkt V')\in\mathcal{S}$ (cf.\ Table \ref{tab2:loop}) are all shaky as they fulfill the equation 
$u_1v_1w_2 - u_2v_2w_1=0$
indicating that the Pl\"ucker coordinates of the four lines $A_iB_i$ are linearly dependent (cf.\ \cite{pottmann}). Note that for the interpretation as 
bar-joint framework or panel-hinge framework the singularity condition consists of a second factor $u_1v_1w_2 + u_2v_2w_1=0$ which implies the 
coplanarity of the vertices $A_i,B_i,A_{i+1},B_{i+1}$ (mod 4) for  $i=1,\ldots, 4$. 
According to Theorem \ref{thm:ident} we get $s(\Vkt L)=s(\Vkt V_1)=s(\Vkt V_2)=\varsigma(\Vkt V_1)=\varsigma(\Vkt V_2)=\varsigma(\Vkt L)$.

Concerning the interpretation of the chain as  bar-joint framework we can give the maximal absolute and relative variation of 
a bar in length during the deformation which equals $0.002359150067$ and $0.001715578691$, respectively.
Moreover, each edge must  change its length in average absolutely by $0.001064975188$ and 
relatively by $0.000993544934$.

\begin{rmk}
In the case of interpreting the chain elements as panel-hinge frameworks or as tetrahedra,
we also put $\nu$ as a unknown in the optimization process as stated in Remark \ref{rem:nu}. 
Computing the critical points with Bertini resulted in $24\,576$ paths. 
For panel-hinge frameworks none of the local extrema for $0\leq\nu<1/2$ has a value less than 
$9.363722223978\cdot 10^{-6}$ and for tetrahedra no local extrema exists within this interval. 
Thus in both cases we get the local minimum at the boundary $\nu=1/2$. \hfill $\diamond$
\end{rmk}

We close this example with some remarks on the case, where the chain is assembled by four tetrahedra 
of Poisson ratio $\nu=1/2$:
Due to the symmetry of the chain we only have to consider one isochoric constraint $\Vol_T'^2-\Vol_T^2=0$, where $T$ 
stands for one of the tetrahedra with vertices $A_i,B_i,A_{i+1},B_{i+1}$ (mod 4) for  $i=1,\ldots, 4$.
The computation of critical points of $u(\Vkt V)$ under this constraint (cf.\ Eq.\ (\ref{extendlagrange2}))
results in the tracking of $279\,936$ paths (cf.\ Table \ref{tab0:loop}). 
If we also invest the information, that the orientation of the tetrahedra has to remain constant under an 
isochoric deformation (cf.\ footnote 14), we can use the simplified condition 
$\Vol'_T-\Vol_T=0$ of oriented volumes, which has 
the half degree. This approach reduces the number of paths to $2\,187$ for the computation of $G(\Vkt V')$ given in 
Table \ref{tab2:loop}. 
But we are still lacking for an efficient determination of an isochoric deformation from $G(\Vkt V)$ into $G(\Vkt V')$ with a monotonically 
increasing deformation energy density. Until now we can only achieve such 
a deformation by a projected gradient descent approach resulting in Fig.\ \ref{fig:4rloop_isochoric}.

\begin{table}
\begin{center}
\begin{footnotesize}
\begin{tabular}[h]{ll|ccc}
 Chain elements & &   $u_1, v_2$ & $v_1, u_2$  & $w_1, w_2$  \\ \hline\hline
	bar-joint & & 0.733113570223  & 0.186762548180  & 0.226463240099  \\ \hline
	panel-hinge &$\nu=1/2$ & 0.733944401620  & 0.187601517133  & 0.226592028651  \\ 
	 &$\nu=1/4$ & 0.733918350832  & 0.187397874243 & 0.226701242326  \\ 
	 &$\nu=0$ & 0.733898439861 & 0.187273478928  & 0.226766038698 \\ \hline
	tetrahedra &  $\nu=1/2$ & 0.735389875237  & 0.190335899873  &  0.225439236330 \\ 
	 &$\nu=1/4$ & 0.735346467296  & 0.190324665294  &  0.225425929220 \\ 
	&$\nu=0$ & 0.735317448771  &  0.190317154629 &  0.225417033375 \\ \hline
\end{tabular}
\end{footnotesize}
\end{center}
\caption{Coordinates of the shaky saddle realization $G(\Vkt V')\in\mathcal{S}$ for the different interpretations of 
the tetrahedra.} \label{tab2:loop}
\end{table}

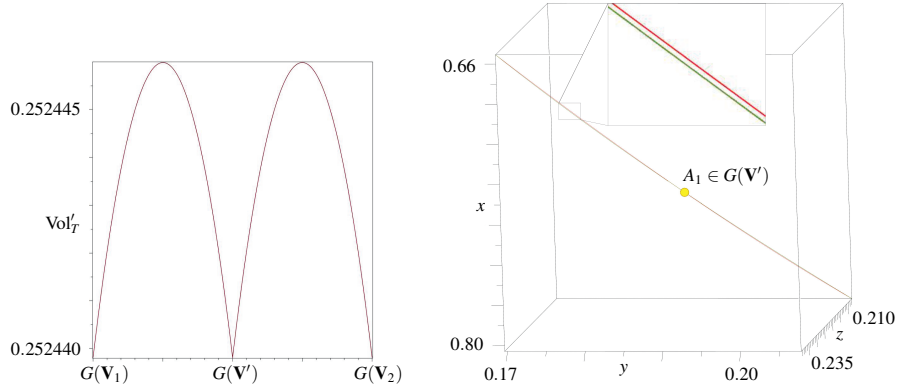
\begin{figure}[t]
\begin{center} 
\begin{overpic}
    [height=40mm]{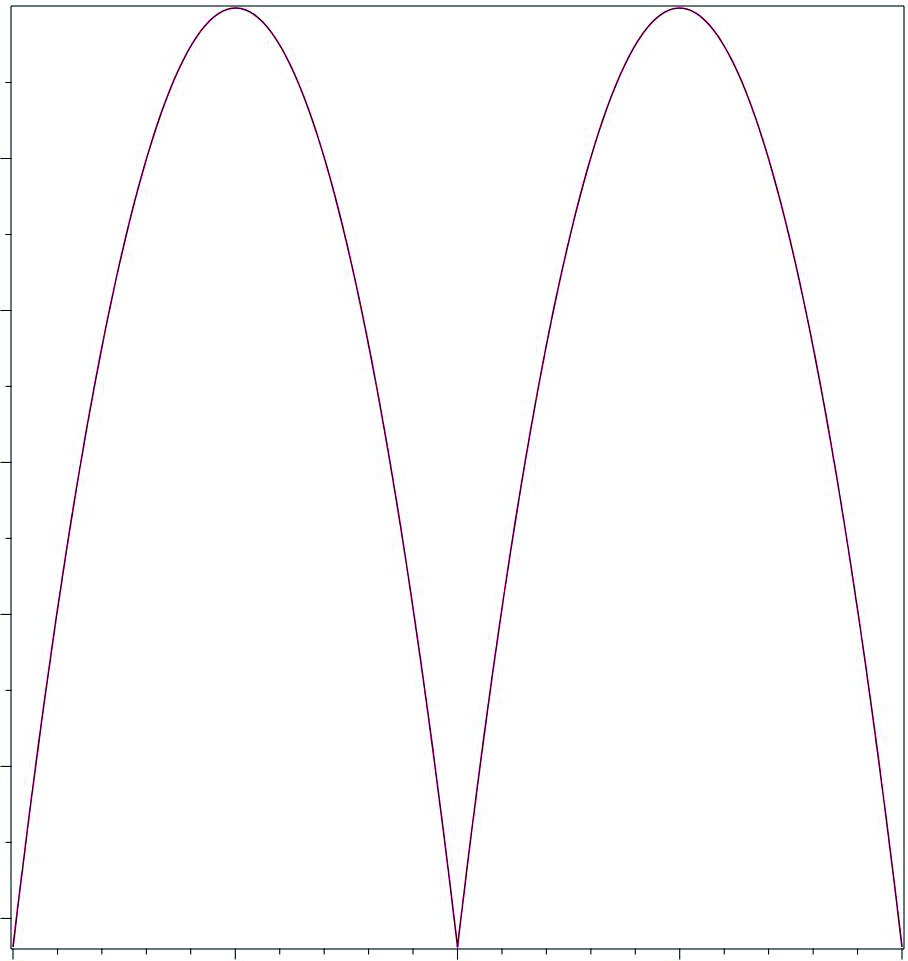}  
\begin{scriptsize}
\put(-26,82){$0.252445$}
\put(-26,2){$0.252440$}
\put(-14.5,44){$\Vol_T'$}
\put(-3,-6){$G(\Vkt V_1)$}
\put(41,-6){$G(\Vkt V')$}
\put(88,-6){$G(\Vkt V_2)$}
\end{scriptsize} 		
  \end{overpic} 
\qquad	\qquad	
	\begin{overpic}
    [height=50mm]{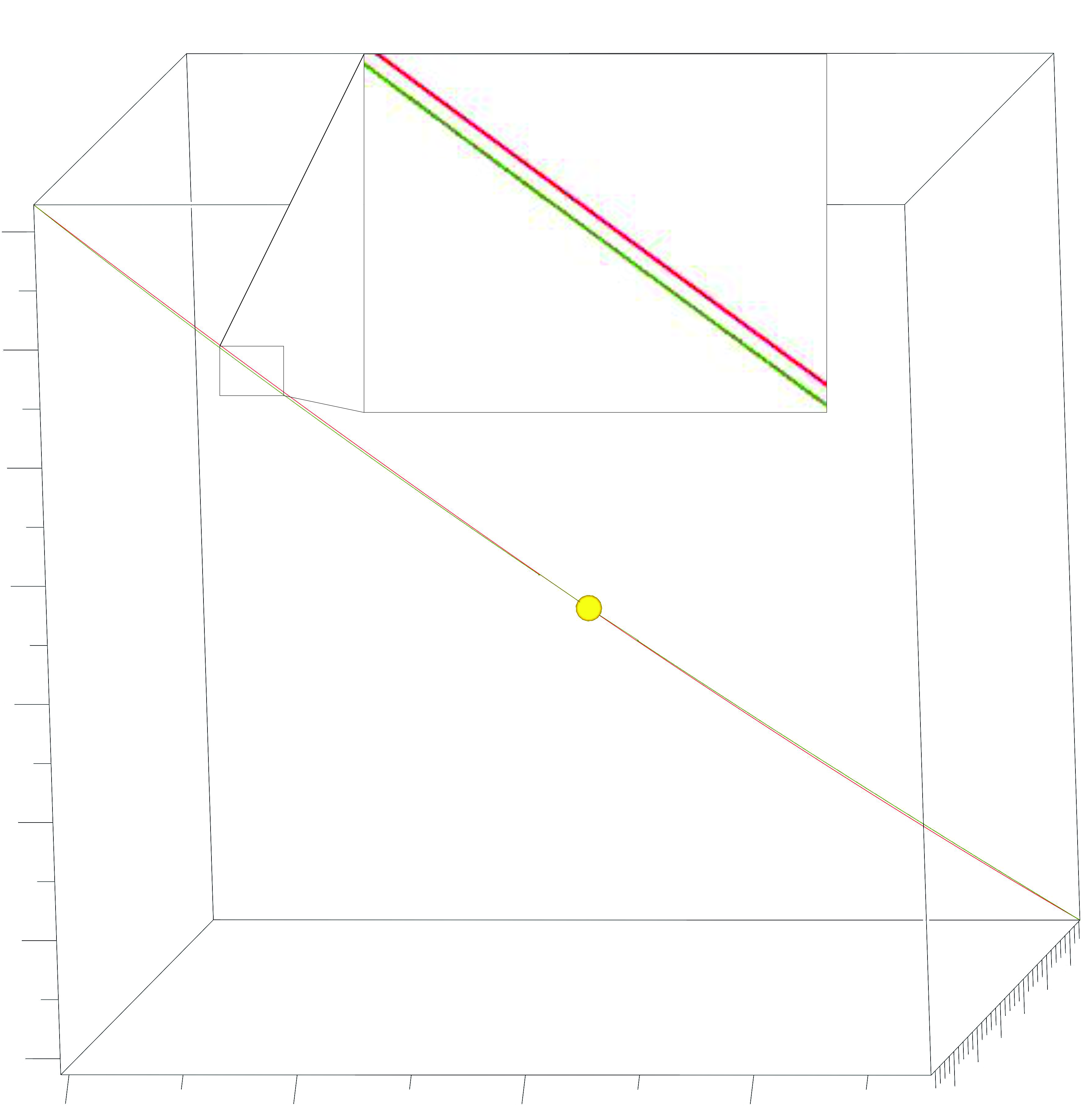}  
\begin{scriptsize}
\put(-11,77){$0.66$}
\put(-9,3){$0.80$}
\put(-2,40){$x$}
\put(0,-5){$0.17$}
\put(36,-3){$y$}
\put(64,-5){$0.20$}
\put(87,-2){$0.235$}
\put(98,10){$0.210$}
\put(93.5,6){$z$}
\put(53,48){$A_1\in G(\Vkt V')$}

\end{scriptsize} 		
  \end{overpic} 
	\end{center}
\caption{Left: The change of the volume of a tetrahedron under the transformation 
implied by the gradient flow in the space of squared leg lengths (cf.\ Eq.\ (\ref{imply}) and 
Remark \ref{rem:notconstvol}). 
The corresponding trajectory of the point $A_1$ is illustrated by the red curve in the right figure. 
The green curve corresponds to an isochoric deformation, which was computed with a projected gradient algorithm.
}
\label{fig:4rloop_isochoric}	
\end{figure}
\end{example}

Further demonstration/verification of our method is done  in Appendix \ref{appendix}, where two snapping model flexors are studied. 
Moreover, the obtained results are compared with those reported in the literature.


\section{Singularity-distance computation for Stewart-Gough manipulators}\label{sec:SGplatform}

A Stewart-Gough (SG) manipulator is a parallel robot consisting of a moving platform, which is connected over six 
telescopic legs to the base. These legs are anchored by spherical joints to the platform and the base (cf.\ Fig.\ \ref{fig:SG1}). 
If the prismatic joints of the legs are fixed, then the pin-jointed body-bar framework is in general rigid. It 
is well-known, that it has an infinitesimal flexibility if and only if the carrier lines of the six legs 
belong to a linear line complex \cite{Merlet}.

A detailed literature review on works dealing with the determination of the closest 
singular configuration to a given non-singular one, which is of interest for singularity-free path-planning and performance optimization of the robot,  
was done  by the author in \cite{WC_2019}. 
Most of these approaches (also the one presented in  \cite{WC_2019}) evaluate the closeness extrinsically (i.e. in the 
6 dimensional configuration space) and not intrinsically (i.e.\ in the 6 dimensional space of prismatic joints).  
Up to the knowledge of the author only one work of Zein et al.\ \cite{zein} determines a singularity-free cube in the 
joint space of a 3-RPR manipulator, which is the planar analogue of a SG platform. For a detailed comparison of 
extrinsic and intrinsic singularity distance measures for planar 3-RPR manipulators we refer to \cite{adit2}.

In the following we want to compute the singularity-distance within the 6-dimensional joint space of the 
manipulator. As a SG manipulator is an isostatic body-bar framework, this computation can be based on Theorem \ref{thm:neu} 
under the additional condition that 
the affine deformations of the platform and the base are restricted to direct isometries.
In this case the function of the strain energy density of a SG manipulator simplifies to:
\begin{equation}\label{eq:rellengchang}
u(\Vkt L')=\frac{1}{\sum_{i=1}^6 L_{i}}\sum_{i=1}^6\frac{({L'_{i}}^2-L_{i}^2)^2}{8L_{i}^3}  
\end{equation}
where $L_i$ (resp.\ $L_i'$)  denotes the length of the undeformed (resp.\ deformed) $i$th leg spanned by 
the platform anchor point $V_i$ and the corresponding base anchor point $V_{i+6}$. 
The base can be pinned down\footnote{In this context it should be mentioned that the results of the paper also hold for pinned 
frameworks  (cf.\ \cite[Section 3.3]{ark2020}).}, i.e. $\Vkt v'_i=\Vkt v_i$ for $i=7,\ldots, 12$, 
and for the platform we set up an affine moving frame with origin $V_1$ and the three vectors $\Vkt v_2-\Vkt v_1$, $\Vkt v_3-\Vkt v_1$ 
and  $(\Vkt v_2-\Vkt v_1)\times (\Vkt v_3-\Vkt v_1)$ under the assumption that $V_1,V_2,V_3$ are not collinear. 
Then one can compute the affine coordinates $(\xi_j,\upsilon_j,\zeta_j)$ of the  points $V_j$ with respect to this frame for $j=4,5,6$, which can 
be used for writting down the coordinate vector of $V'_j$ as follows:
\begin{equation}\label{eq:affkombi}
\Vkt v'_j=\Vkt v_1' + \xi_j(\Vkt v'_2-\Vkt v'_1) + \upsilon_j(\Vkt v'_3-\Vkt v'_1) + \zeta_j[(\Vkt v'_2-\Vkt v'_1)\times (\Vkt v'_3-\Vkt v'_1)].
\end{equation}
This affine transformation is an orientation preserving isometry if the three side conditions $e_i=0$ with
\begin{equation}
e_1:=\|\Vkt v'_2-\Vkt v'_3\|^2-\|\Vkt v_2-\Vkt v_3\|^2,\quad
e_2:=\|\Vkt v'_3-\Vkt v'_1\|^2-\|\Vkt v_3-\Vkt v_1\|^2\quad\text{and}\quad
e_3:=\|\Vkt v'_1-\Vkt v'_2\|^2-\|\Vkt v_1-\Vkt v_2\|^2
\end{equation}
hold true. 
Under consideration of Eq.\ (\ref{eq:affkombi}) one can write Eq.\ (\ref{eq:rellengchang}) in dependence of $\Vkt v'_1,\Vkt v'_2 ,\Vkt v'_3$, 
which is part of the Lagrangian 
\begin{equation}\label{eq:finlag}
F(\Vkt v'_1,\Vkt v'_2 ,\Vkt v'_3,\eta_1,\eta_2,\eta_3)= u(\Vkt v'_1,\Vkt v'_2 ,\Vkt v'_3) -\eta_1e_1-\eta_2e_2-\eta_3e_3.
\end{equation}

\begin{rmk}
By using the linear combination given in Eq.\ (\ref{eq:affkombi}) instead of the formulation of Eq.\ (\ref{us:lincomb}) the 
number of side conditions forcing an isometric transformation is reduced from 6 to 3; i.e.\  $e_1=e_2=e_3=0$. 
In addition this formulation already restricts to direct isometries.
Clearly, one can also use a parametrization of the Euclidean motion group (e.g.\ Study parameters) for the representation of $\Vkt v'_1,\ldots, \Vkt v'_6$. 
We use here this so-called point-based formulation as it turned out to have certain computational advantages (cf.\ \cite{adit}). \hfill $\diamond$
\end{rmk}

The formulation given in Eq.\ (\ref{eq:rellengchang}) rely on the change of the leg lengths relative to its initial length. As we are now working in the 
joint space of the SG manipulator also the following function of absolute changes in the leg lengths makes sense:
\begin{equation}\label{eq:abslengchang}
l(\Vkt L')=\sum_{i=1}^6({L'_{i}}^2-L_{i}^2)^2.
\end{equation}
Therefore one can also substitute $u(\Vkt v'_1,\Vkt v'_2 ,\Vkt v'_3)$ by  $l(\Vkt v'_1,\Vkt v'_2 ,\Vkt v'_3)$ in Eq.\ (\ref{eq:finlag}). 
For both Lagrangians, analogous considerations as done in the proof of Theorem \ref{thm:neu} show that the saddle realizations have to be 
shaky (as the bodies cannot reduce in dimension due to the enforced direct isometries). We demonstrate this in the following example.

\begin{example}
We study a SG manipulator, which is of interest for practical applications, as the positioning and orientation of the 
relative pose of the platform and the base is decoupled (cf.\ \cite[Section VI]{borras}). 
The moving platform has a semihexagonal shape (with central angles of $\pi/6$ and $\pi/2$, respectively) and the base is a truncated triangular pyramid 
(cf.\ Fig.\ \ref{fig:SG1}). The coordinates of the base anchor points $V_7,\ldots ,V_{12}$ with respect to the fixed frame are given by:

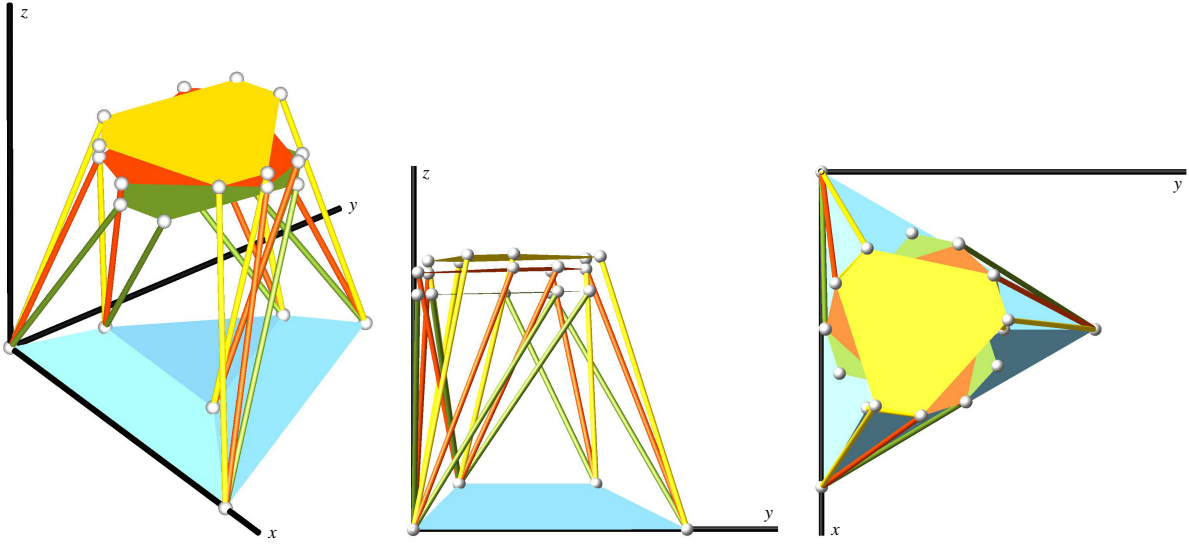
\begin{figure}
\begin{center} 
\begin{overpic}
    [width=50mm]{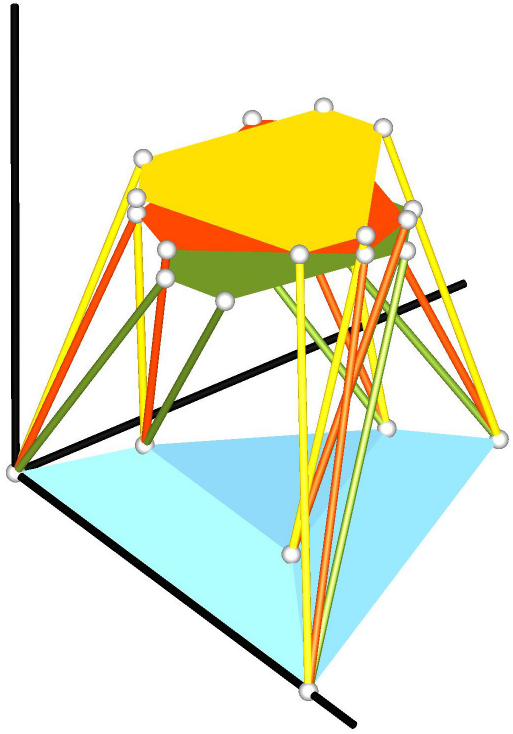}  
\begin{scriptsize}
\put(50,0){$x$}
\put(65,61){$y$}
\put(4,97){$z$}
\end{scriptsize} 		
  \end{overpic} 
	\quad
\begin{overpic}
    [width=50mm]{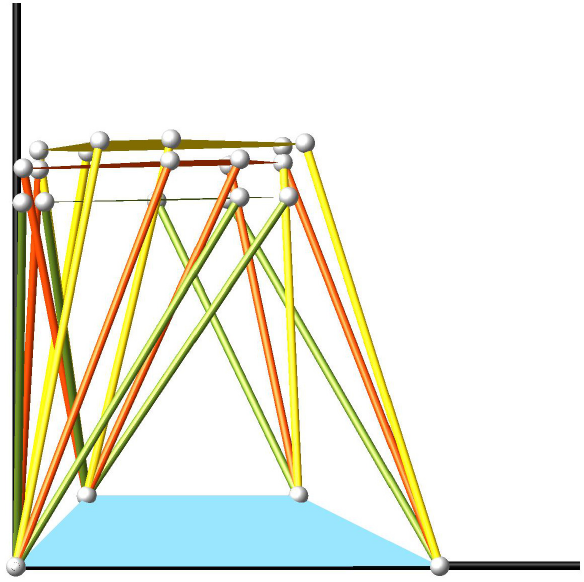}  
\begin{scriptsize}
\put(96,5.5){$y$}
\put(5,96){$z$}
\end{scriptsize} 		
  \end{overpic} 
	\quad
\begin{overpic}
    [width=50mm]{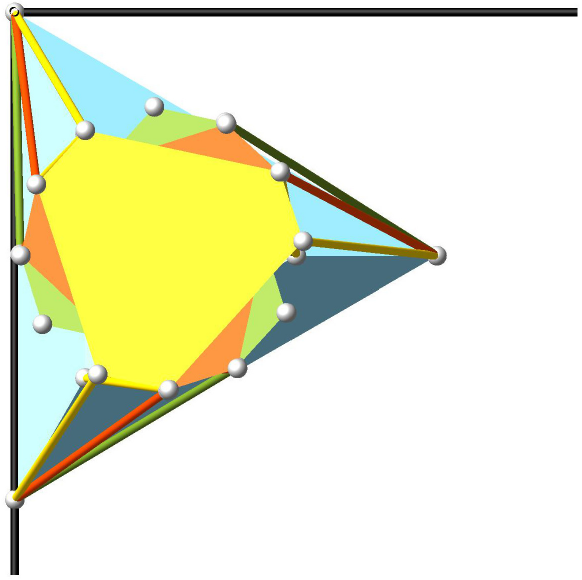}  
\begin{scriptsize}
\put(96,93){$y$}
\put(5,1){$x$}
\end{scriptsize} 		
  \end{overpic} 
\end{center} 
\hfill
\begin{scriptsize}
\end{scriptsize}
\caption{
Axonometric view (left), front view (center) and top view (right):
$G(\Vkt V_1)$ is illustrated in yellow, $G(\Vkt V_2)$ is displayed in green and $G(\Vkt V')$ w.r.t.\ $u(\Vkt L')$ is shown in red. 
Note that the realization $G(\Vkt V')$ w.r.t.\ $l(\Vkt L')$ is not shown here as it is too close (cf.\ Table \ref{tab1:sgp2}) to $G(\Vkt V')$ w.r.t.\ $u(\Vkt L')$. 
}\label{fig:SG1}	
\end{figure}

\begin{equation}
\Vkt v_7:=\begin{pmatrix} 0 \\ 0 \\ 0 \end{pmatrix}, \quad
\Vkt v_8:=\begin{pmatrix} \tfrac{\sqrt{3}}{2} \\ \tfrac{1}{2} \\ \tfrac{1}{2} \end{pmatrix}, \quad
\Vkt v_9:=\begin{pmatrix} 2\sqrt{3} \\ 0 \\ 0 \end{pmatrix}, \quad
\Vkt v_{10}:=\begin{pmatrix} \tfrac{3\sqrt{3}}{2} \\ \tfrac{1}{2} \\ \tfrac{1}{2} \end{pmatrix}, \quad
\Vkt v_{11}:=\begin{pmatrix} \sqrt{3} \\ 3 \\ 0 \end{pmatrix}, \quad
\Vkt v_{12}:=\begin{pmatrix} \sqrt{3} \\ 2 \\ \tfrac{1}{2} \end{pmatrix}.
\end{equation}
The vertices $V_1,\ldots ,V_{6}$ of the moving platform are determined by the three conditions
\begin{equation}
\|\Vkt v_2-\Vkt v_3\|^2=2,\quad
\|\Vkt v_3-\Vkt v_1\|^2=3,\quad
\|\Vkt v_1-\Vkt v_2\|^2=2-\sqrt{3}
\end{equation}
and the affine coordinates $(\xi_j,\upsilon_j,0)$  for $j=4,5,6$  given by
\begin{equation}
\xi_4:=-\tfrac{\sqrt{3}+1}{2},\quad \upsilon_4 := \tfrac{\sqrt{3}+1}{2}, \quad
\xi_5:=-\tfrac{3(\sqrt{3}+1)}{2}, \quad \upsilon_5 := \tfrac{\sqrt{3}+1}{2}, \quad
\xi_6:= -\sqrt{3}-2,\quad \upsilon_6 := 1.
\end{equation}
The input data is completed by the following six leg lengths:
\begin{equation}
L_1:=31/10, \quad 
L_2:=25/10,\quad 
L_3:=32/10,\quad 
L_4:=26/10,\quad
L_5:=315/100, \quad
L_6:=255/100.
\end{equation}

We compute the closest singularity with respect to both intrinsic metrics in the 6-dimensional joint space, which are given in 
Eq.\ (\ref{eq:rellengchang}) and Eq.\ (\ref{eq:abslengchang}), respectively. 
The computational data is summarized in Table \ref{tab1:sgp1}.

\begin{table}[h!]
\begin{center}
\begin{footnotesize}
\begin{tabular}[h]{l|cccccc}
Intrinsic metric  & \# tracked paths & $\# \mathcal{R}$ & $\# \mathcal{S}$ & $\varsigma(\Vkt L) =s(\Vkt L)$ \\ \hline\hline
$u(\Vkt L')$ 		& $26\,017$  & $124$  & $111$ &   $3.324106490339\cdot 10^{-5} $ \\ \hline 
$l(\Vkt L')$ 		& $25\,473$  & $122$  & $112$ &   $1.024890249080\cdot 10^{-1}$ \\ \hline 
\end{tabular}
\end{footnotesize}
\end{center}
\caption{Computational data: Note that the computation of the set $\mathcal{R}$ was done by a 
regeneration homotopy performed with Bertini. 
Note that the set $\mathcal{S}$ has to be filtered out from $\mathcal{R}$ by using a second-derivative test based
on the bordered Hessian \cite{spring} due to the three side conditions $e_1=e_2=e_3=0$.
} \label{tab1:sgp1}
\end{table}

For this input data the SG manipulator has 4 real solutions for the direct kinematics, whereby two solutions are out of 
interest as the platform is below the base. The other two undeformed realizations are illustrated in Fig.\ \ref{fig:SG1} and 
the corresponding vectors $\Vkt v_i=(x_i,y_i,z_i)^T$ (with respect 
to the fixed frame) of the platform anchor points $V_i$ for $i=1,2,3$ are given in Table \ref{tab1:sgp2}. 
This table also contains the corresponding coordinate entries of the two saddle realization  $G(\Vkt V')$, which 
imply  the singularity-distance with respect to the different metrics. 
It can easily be checked that in both realizations $G(\Vkt V')$ the lines of the six legs belong to a linear line complex \cite{Merlet}.

\begin{table}[h!]
\begin{center}
\begin{footnotesize}
\begin{tabular}[h]{c|cccc}
 & $\Vkt V_1$ & $\Vkt V_2$ & $\Vkt V'$ w.r.t.\ $u(\Vkt L')$ & $\Vkt V'$ w.r.t.\ $l(\Vkt L')$ \\ \hline\hline
 $x_1$ &  0.842928302224 	 	& 1.722185861193  	&  1.225741950663  & 1.221043138727  \\ \hline
 $x_2$ & 	1.227117667465  	& 2.215337180628  	&  1.731161637163  & 1.726370305617  \\ \hline
 $x_3$ & 	2.571092053711  	& 2.516815584869  	&  2.684588070901  & 2.680372134138  \\ \hline
 $y_1$ & 	0.505604311407 		& 0.043791592870  	&  0.160342690355  & 0.156045184954  \\ \hline
 $y_2$ & 	0.158839883661  	& 0.201000858188 		&  0.048866028718  & 0.044256807218  \\ \hline
 $y_3$ & 	0.594205901934  	& 1.582339155201  	&  1.092093613851  & 1.086968370169  \\ \hline
 $z_1$ & 	2.940040162536  	& 2.577238474600  	&  2.814950790353  & 2.824067090458  \\ \hline
 $z_2$ & 	2.950147373651  	& 2.583256403538  	&  2.823499874485  & 2.833916224311  \\ \hline
 $z_3$ & 	3.014872015394  	& 2.615119877650  	&  2.875019201615  & 2.885230039892  \\ \hline 
\end{tabular}
\end{footnotesize}
\end{center}
\caption{Coordinates of the undeformed realizations $G(\Vkt V_1)$ and  $G(\Vkt V_2)$ of the SG manipulator and of the 
shaky saddle realization  $G(\Vkt V')$ with respect to both
intrinsic metrics given in 
Eq.\ (\ref{eq:rellengchang}) and Eq.\ (\ref{eq:abslengchang}), respectively.}\label{tab1:sgp2}
\end{table}
\end{example}

\begin{rmk}
Note that the snapping octahedra of Wunderlich \cite{wunderlich_achtflach} imply further examples of octahedral 
hexapods with a high snapping capability.  
\hfill $\diamond$
\end{rmk}


\section{Conclusion}\label{sec:conclusion}

The first considerations in the design process of a pin-jointed body-bar framework  concern its geometry. 
In this paper we presented an index, which evaluates the framework geometry with respect to its capability to snap. 
As this so-called {\it snappability} only depends on the intrinsic framework geometry, it enables a fair comparison of 
frameworks differing in the combinatorial structure, inner metric and types of structural elements 
(bars, polygonal panels or polyhedral bodies).  Therefore it can serve 
engineers as a criterion in an early design stage for the geometric layout of frameworks without or with 
the capability to snap, depending on the application. In the context of multistability the index can for example be used 
for the geometric layouting of unit-cells/building-blocks of periodic 
metamaterials (e.g.\ \cite{haghpanah,shang,yang,danso,klein}) and origami structures (e.g.\ hypar tessellation\cite{liu}, waterbomb cylinder tessellation \cite{gillman} or the Kresling pattern 
with a circular arrangement to closed strips, which correspond to snapping antiprisms \cite{wunderlich_antiprism} and can be composed to cylindrical towers \cite{liu2017,kresling,cai}, 
or a helical  arrangement according to C.R.\ Calladine studied in \cite{guestII,guestI,guestIII,wittenburg}). 
For the resulting framework geometry one can specify in a later design phase the material and 
dimensioning (e.g.\ profile of bars)  of the framework elements in such a way that the wanted effects are even increased. 

We also demonstrated on basis of parallel manipulators of Stewart-Gough type, that our approach can be used for the computation of intrinsic singular-distances.
The computational aspects of the index were also illuminated and an allover algorithm  was presented
for isostatic frameworks (under affine deformation of the bodies). 
Note that this algorithm does not take the bending of panels into account and also 
ignores the collision of bars and/or bodies during the deformation.
For overbraced frameworks we gave a strategy for the snappability computation but the result is 
without guarantee. In this case one can use the proposed lower bounds, which always hold true. \bigskip

\noindent
{\bf Acknowledgments.} 
The  author  is supported by grant P\,30855-N32 of the Austrian Science Fund FWF 
as well as by project F77 (SFB ``Advanced Computational Design'', subproject SP7).
Moreover, the author thanks Miranda Holmes-Cerfon for providing the data of the four-horn example 
discussed in Appendix \ref{sec:FH}. Further thanks to Aditya Kapilavai for his help in outsourcing some of the 
Bertini computations.


\appendix

\section{Examples of  snapping model flexors}\label{appendix}

Within this appendix we discuss the examples of the Siamese dipyramid and the four-horn in detail and compare the 
obtained results with existing ones reported in the literature. 

\subsection{Siamese dipyramid}\label{sec:SD}

The original Siamese dipyramid (SD) introduced by Goldberg \cite[page 167]{goldberg} consists of 20 equilateral triangles with an edge lengths of 1, 
which are arranged in two dipyramids with a hexagonal equatorial polygon (see Fig.\ \ref{fig:1_SD}, left). 
Note that we assume that the SD has a reflexion-symmetry with respect to two orthogonal planes. 
We can insert a coordinate frame in such a way that these planes are the $xy$-plane and the $yz$-plane, respectively.

\begin{figure}[t]
\begin{center} 
\begin{minipage}{110mm}
\begin{overpic}
    [width=55mm]{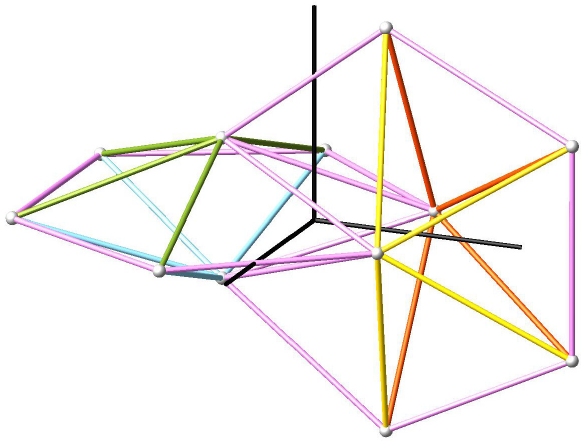}  
\begin{scriptsize}
\put(0.5,33){$A_1$}
\put(10.7,50.5){$\overline{A}_1$}
\put(34.5,22){$\overline{C}_2$}
\put(23,25){$B_1$}
\put(35,55){${C}_2$}
\put(56,52){$\overline{B}_1$}
\put(55.5,72){$z$}
\put(47,35){$x$}
\put(91,33){$y$}
\put(69,71){$B_2$}
\put(59.2,0){$\overline{B}_2$}
\put(97,7.7){$\overline{A}_2$}
\put(97,53){$A_2$}
\put(59.1,27.2){${C}_1$}
\put(74.1,43.2){$\overline{C}_1$}
\end{scriptsize} 		
  \end{overpic} 
	\quad
\begin{overpic}
    [width=55mm]{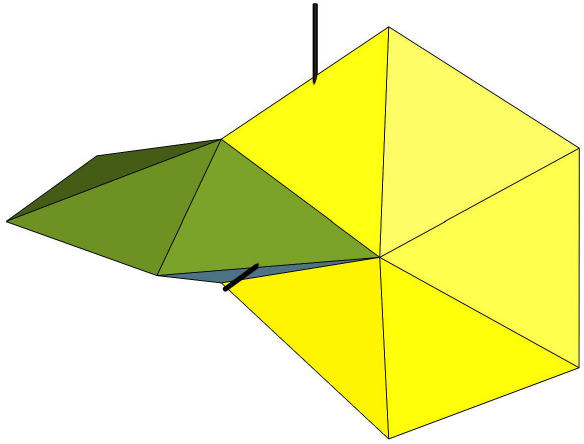}  
\begin{scriptsize}
\end{scriptsize} 		
  \end{overpic} 
\end{minipage}\hfill	
\begin{minipage}{40mm}
\begin{overpic}
    [width=32mm]{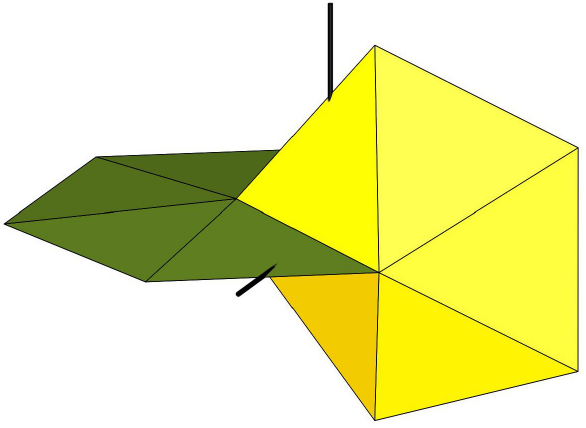}  
  \end{overpic}
\begin{overpic}
    [width=32mm]{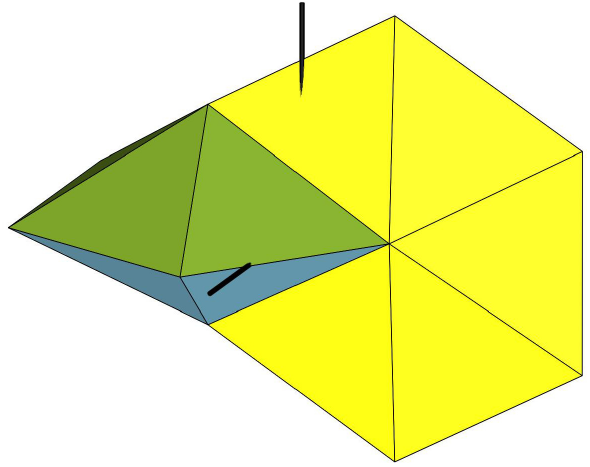}  
  \end{overpic}	
\end{minipage}	
\end{center} 
\hfill
\begin{scriptsize}
\end{scriptsize}
\caption{Left: Illustration of the realization $G(\Vkt V_1)$  as a bar-joint framework together with the coordinate frame, 
where the axes are of length one. Center: The same configuration as on the left side but illustrated with panels instead of bars. 
Right: At the top the second realization $G(\Vkt V_2)$ is visualized and at the bottom the third one $G(\Vkt V_3)$.
}\label{fig:1_SD}	
\end{figure}

\begin{figure}[t]
\begin{center} 
\begin{overpic}
    [width=55mm]{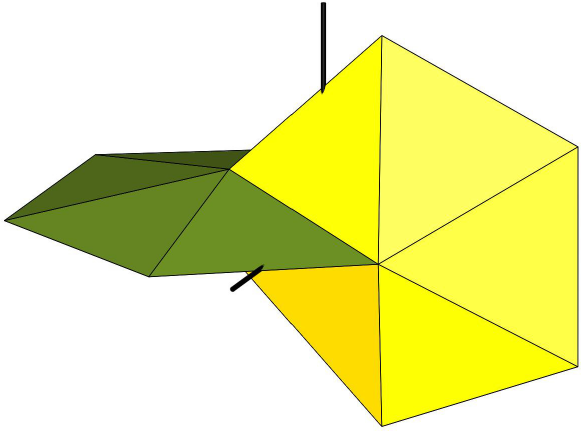}  
\begin{scriptsize}
\end{scriptsize} 		
  \end{overpic} 
	\qquad
\begin{overpic}
    [width=55mm]{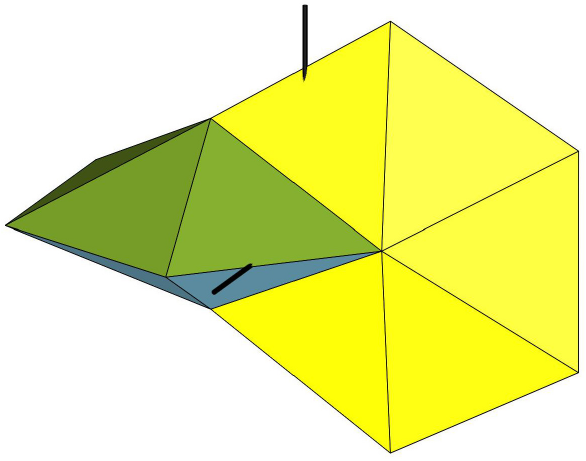}  
\begin{scriptsize}
\end{scriptsize} 		
  \end{overpic} 
\end{center} 
\hfill
\begin{scriptsize}
\end{scriptsize}
\caption{On the left (resp.\ right) side the shaky saddle realizations $G(\Vkt V')$
(resp.\ $G(\Vkt V'')$) is illustrated, which is passed during the snap between $G(\Vkt V_1)$ and 
$G(\Vkt V_2)$ (resp.\  $G(\Vkt V_3)$).  
An animation of the snapping behavior can be downloaded from \cite{data2021}.  
}\label{fig:2_SD}	
\end{figure}

Then the vertices, which are noted according to Fig.\ \ref{fig:1_SD}, can be coordinatized as follows: 
\begin{align}
A_1&=(x_1,y_1,0)^T &\quad \overline{A}_1&=(-x_1,y_1,0)^T &\quad A_2&=(0,u_1,v_1)^T &\quad \overline{A}_2&=(0,u_1,-v_1)^T \\
B_1&=(x_2,y_2,0)^T &\quad \overline{B}_1&=(-x_2,y_2,0)^T &\quad B_2&=(0,u_2,v_2)^T &\quad \overline{B}_2&=(0,u_2,-v_2)^T \\
C_1&=(x_3,y_3,0)^T &\quad \overline{C}_1&=(-x_3,y_3,0)^T &\quad C_2&=(0,u_3,v_3)^T &\quad \overline{C}_2&=(0,u_3,-v_3)^T 
\end{align}
In addition we can assume without loss of generality
that $u_3=-y_3$ holds; i.e. the vertices $C_1$ and $\overline{C}_1$ have the same 
distance from the $xz$-plane as the points $C_2$ and $\overline{C}_2$. 
Therefore the total number of unknowns is 11. 

It is well-known \cite{goldberg,gorkavyy}  that the SD can snap out of the symmetric\footnote{With respect to the height of the two dipyramids.} 
realization $G(\Vkt V_1)$ (cf.\ Fig.\ \ref{fig:1_SD}, left/center)  into one of the two asymmetric realizations $G(\Vkt V_2)$ 
and $G(\Vkt V_3)$, respectively (cf.\ Fig.\ \ref{fig:1_SD}, right). 
A simple procedure for the computation of these three undeformed realizations is given in \cite{gorkavyy}. 
We only give the numerical values of these configurations in Table \ref{tab:siam}.

\begin{table}
\begin{center}
\begin{footnotesize}
\begin{tabular}[h]{c|cccc}
 & $\Vkt V_1$ & $\Vkt V_2$ & $\Vkt V'$ (bar-joint) & $\Vkt V'$ (panel-hinge) \\ \hline\hline
 $x_1$ & -0.5							 & -0.5								& -0.501499108259  	& -0.501518680610  \\ \hline
 $x_2$ & -0.940024410925 	 & -0.997453425271 		& -0.979262620688 	& -0.979200605399  \\ \hline
 $x_3$ & -0.327267375345	 & -0.492373245899 		& -0.432379113707 	& -0.432385909548  \\ \hline
 $y_1$ & -1.245032582350	 & -1.296828963170 		& -1.282364611843 	& -1.282380966624  \\ \hline
 $y_2$ & -0.347046770776	 & -0.429338277522 		& -0.400464708308  	& -0.400381868195  \\ \hline
 $y_3$ &  0.443224584739		& 0.433734148410 		&  0.440320490435 	&  0.440337472811  \\ \hline
 $u_1$ &  1.245032582350		& 1.146172627664 		&  1.193026874842 	&  1.192998833484  \\ \hline
 $u_2$ &  0.347046770776		& 0.205744933405 		&  0.273852988315 	&  0.273969061570  \\ \hline
 $v_1$ &  0.5								& 0.5								&  0.498976790866 	&  0.498943974578  \\ \hline
 $v_2$ &  0.940024410925		& 0.839993752693 		&  0.887603513697 	&  0.887698700553  \\ \hline
 $v_3$ &  0.327267375345		& 0.071185256433 		&  0.190566289070 	&  0.190551925851  \\ \hline
\end{tabular}
\end{footnotesize}
\end{center}
\caption{Coordinates of the undeformed realizations $G(\Vkt V_1)$ and  $G(\Vkt V_2)$, respectively, and of the 
shaky saddle realization  $G(\Vkt V')$ with respect to the two different interpretations.  
The coordinates of $G(\Vkt V_3)$ and $G(\Vkt V'')$ can be obtained from $G(\Vkt V_2)$ and $G(\Vkt V')$
by the following exchange of coordinate entries:
$x_i\leftrightarrow -v_i$ and  $y_i\leftrightarrow -u_i$ for $i=1,2,3$. 
}\label{tab:siam}
\end{table}

\subsubsection{Isostaticity and shakiness}\label{sec:isoshaky}
The bar-joint framework of the SD is isostatic, because every closed polyhedral surface of genus 0 with triangular faces
has this property\footnote{This can easily be followed from Euler's polyhedral formula.}. This isostaticity 
remains intact under the assumption of the 2-fold reflexion-symmetry, as it only corresponds to the 
identification of some of the coordinates within the structure. 

The SD is in a shaky configuration if the rank of its rigidity matrix $\Vkt R_{G(\Vkt V)}$ is less than 30. 
From this one can compute the algebraic characterization, which corresponds to the vanishing of the following polynomial

{\begin{equation}
\underbrace{v_3(x_1y_2 + x_1y_3 - x_2y_1 - x_2y_3)}_{\text{copl}(C_2,\overline{C}_2,A_1,B_1)}
\underbrace{v_3(2x_2y_3 - x_3y_2 - x_3y_3)}_{\text{copl}(C_2,\overline{C}_2,B_1,C_1)}
\underbrace{x_3(u_1v_2 - u_2v_1 + v_1y_3 - v_2y_3)}_{\text{copl}(C_1,\overline{C}_1,A_2,B_2)}
\underbrace{x_3(u_2v_3 + 2v_2y_3 - v_3y_3)}_{\text{copl}(C_1,\overline{C}_1,B_2,C_2)}S
\end{equation}}
where $\text{copl}$ indicates the coplanarity of the vertices given in the round bracket. 
For the condition $x_3=0$ or $v_3=0$ one of the two dipyramids is even in a flat configuration. 
Beside these geometric simple cases of shakiness we also have the factor\footnote{It can be downloaded from \cite{data2021}.}
$S$, which denotes an algebraic expression with 374 terms and a total degree of 9. 
Interestingly $S$ is only quadratic with respect to the two non-zero coordinates of the following points: $A_i$, $B_i$, $\overline{A}_i$ and $\overline{B}_i$ for $i=1,2$.  
For the points $A_i$ and $\overline{A}_i$ it is even linear in $x_1$ (for $i=1$) or $v_1$ (for $i=2$). 

Moreover, if the infinitesimal flexibility of the SD interpreted as bar-joint framework does not result from the degeneration of a triangular  substructure into a collinear arrangement, then 
the corresponding panel-hinge framework is also shaky.

\subsubsection{Interpretation as a bar-joint structure}
We set up our formulation of the deformation energy density $u$  
under the assumption that the SD keeps the 2-fold reflexion-symmetry during the deformation. 
The obtained system of $11$ equations $\nabla u$ results in $177\,147$ paths within a total degree homotopy (cf.\ \cite{bates}). 
The path tracking done by the software Bertini ends up in $22\,153$ finite real solutions (set $\mathcal{R}$). 
After reduction to the set $\mathcal{S}$ we remain with $21\,904$ solutions. This set is the input for  
the algorithm described in Section \ref{isostatic}, which outputs the 
two shaky saddle realizations $G(\Vkt V')$ and  $G(\Vkt V'')$, respectively, displayed in Fig.\ \ref{fig:2_SD}. 
The numerical values of these realizations are also given in Table \ref{tab:siam}. 

We get $s(\Vkt L)=s(\Vkt V_1)=s(\Vkt V_2)=s(\Vkt V_3)=1.661376004928\cdot 10^{-6}$ and due to Theorem \ref{thm:ident} (under consideration of 
Corollary \ref{rep:shaky}) this value also equals $\varsigma(\Vkt L)=\varsigma(\Vkt V_1)=\varsigma(\Vkt V_2)=\varsigma(\Vkt V_3)$. \bigskip

\noindent
{\bf Comparison with the results obtained in \cite{almost}:}
According to \cite{almost} there exists a realization  within the deformation path between two snapping realizations 
$G(\Vkt V_1)$ and $G(\Vkt V_{2/3})$, where the value for  $e:=\sqrt{\sum_{ij}(L_{ij}^2-L_{ij}'^2)^2}$ 
is greater or equal to a value $e_{min}^*$ given by $2.98\cdot 10^{-4}$ for $G(\Vkt V_1)$ and $3.35\cdot 10^{-4}$ for $G(\Vkt V_{2/3})$, respectively. 
As noted in \cite{almost} the value $e_{min}^*$ is a minimum bound and does not say how close this bound is to the true barrier. 

By our approach we can determine this true barrier value numerically as $e_{min}=1.996823079751\cdot 10^{-2}$, which has to be the same for the three realizations $G(\Vkt V_i)$ for $i=1,2,3$ due to the snapping between these 
realizations. 
Therefore the true barrier is approximately 67 times and 60 times, respectively, larger then the given $e_{min}^*$ value.

From $e_{min}^*$ one can also approximate the length $\Delta L_{\diameter}^*$, which an edge must change in average according to \cite[Example 2]{almost}, 
yielding the values $2.720355368942 \cdot 10^{-5}$ and $3.058117612737\cdot 10^{-5}$, respectively. 
Based on $G(\Vkt V')$ we can also compute the absolute average change $\Delta L_{\diameter}^{abs}=1.673630072024\cdot 10^{-3}$ 
(which equals also the relative average change $\Delta L_{\diameter}^{rel}$ as the initial length of the edges is 1). Therefore it is
62 times and 55 times, respectively, larger than the values resulting from the data given in  \cite{almost}. \bigskip

\noindent
{\bf Comparison with the results obtained in \cite{gorkavyy}:}
The intrinsic index given in \cite{gorkavyy} equals $\Delta L_{max}^*=3.94\cdot 10^{-3}$ and correspond to the maximal relative (with respect to the initial length of 1) change  in the length of an edge during the deformation. 
But it should be noted that the setup of the pyramids in \cite{gorkavyy} is more restrictive than ours, as all edges through the vertices  $C_1$, $\overline{C}_1$ and $C_2$, $\overline{C}_2$ cannot be deformed and have a fixed length of 1;  all other edges are restricted to have the same length. 

Based on $G(\Vkt V')$ we can also compute this maximal relative change $\Delta L_{max}^{rel}$ of an edge as $2.998216519082\cdot 10^{-3}$ (which equals also the maximal absolute change $\Delta L_{max}^{abs}$ as the initial length of the edges is 1). 
Therefore the value reported in \cite{gorkavyy} is $31\%$ larger than ours.

\begin{rmk} 
In \cite{gorkavyy} and \cite{almost} also some indices are given to estimate/quantify the variations of the spatial shape of the snapping framework. Due to their extrinsic 
nature they cannot give information about the 
snappability, which only depends on the 
intrinsic geometry. \hfill $\diamond$
\end{rmk}

\subsubsection{Interpretation as a panel-hinge structure}
The same study can also be done by considering the SD as a polyhedral surface composed of 
triangular panels with a Poisson ratio of $\nu=1/2$. The tracking of the $177\,147$ paths of a total degree homotopy using Bertini ends up 
in 20305 real solutions, which can be reduced to $20\,056$ solutions of the set $\mathcal{S}$. 
In this case we get $s(\Vkt L)=s(\Vkt V_i)=\varsigma(\Vkt V_i)=\varsigma(\Vkt L)=4.466362657431\cdot 10^{-6}$ for $i=1,2,3$.


\begin{figure}[t]
\begin{center} 
\begin{minipage}{110mm}
\begin{overpic}
    [width=55mm]{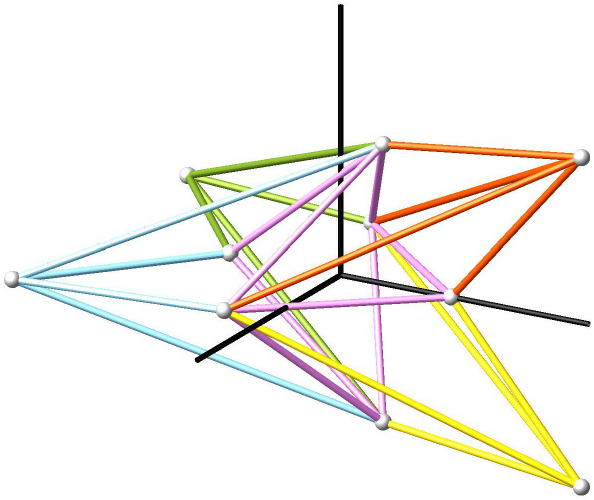}  
\begin{scriptsize}
\put(0.5,31.5){$A_1$}
\put(26,57){$\overline{A}_1$}
\put(98,5.5){$\overline{A}_2$}
\put(98,60.5){${A}_2$}
\put(31,43){$B_1$}
\put(73.3,38.4){$B_2$}
\put(32.5,27.6){$C_1$}
\put(59.5,8){$\overline{C}_2$}
\put(60,62.5){${C}_2$}
\put(59,80){$z$}
\put(65,50.5){$\overline{C}_1$}
\put(97,26){$y$}
\put(30,21){$x$}
\end{scriptsize} 		
  \end{overpic} 
	\quad
\begin{overpic}
    [width=55mm]{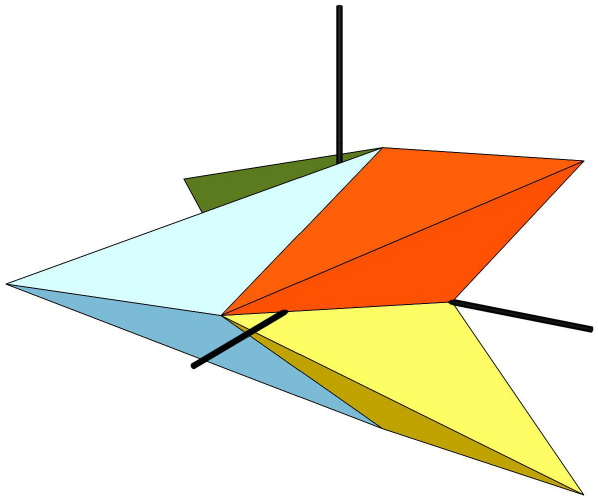}  
\begin{scriptsize}
\end{scriptsize} 		
  \end{overpic} 
\end{minipage}\hfill	
\begin{minipage}{40mm}
\begin{overpic}
    [width=35
		mm]{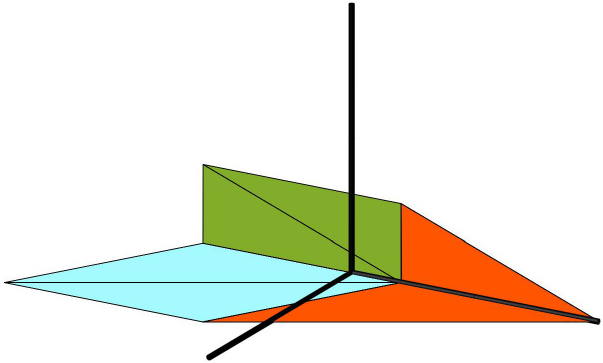}  
  \end{overpic}
\begin{overpic}
    [width=28mm]{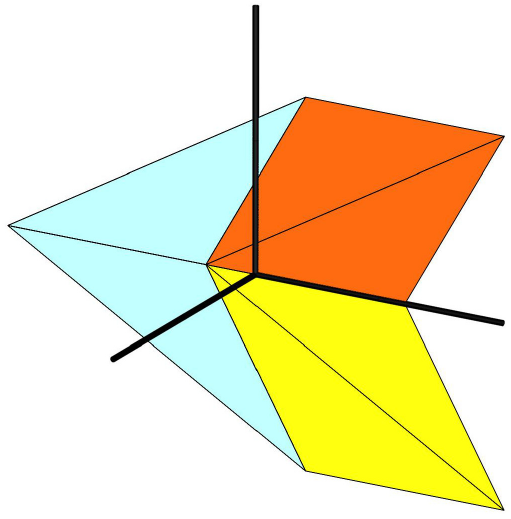}  
  \end{overpic}	
\end{minipage}	
\end{center} 
\hfill
\begin{scriptsize}
\end{scriptsize}
\caption{Left: Illustration of the realization $G(\Vkt V_1)$ of the original four-horn as a bar-joint framework together with the coordinate frame, 
where the axes are of length one. Center: The same configuration as on the left side but illustrated with panels instead of bars. 
Right: The two flat realizations are visualized, where $G(\Vkt V_2)$ is displayed at the top and $G(\Vkt V_3)$  at the bottom. 
}\label{fig:1_FH}	
\end{figure}

\begin{figure}[t]
\begin{center} 
\begin{overpic}
    [width=55mm]{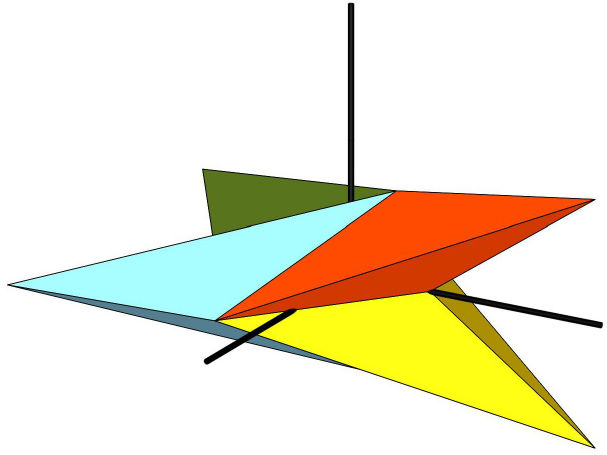}  
\begin{scriptsize}
\end{scriptsize} 		
  \end{overpic} 
	\qquad
\begin{overpic}
    [width=55mm]{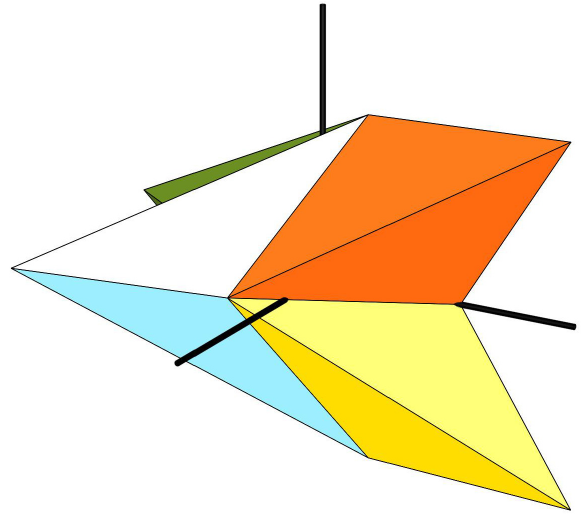}  
\begin{scriptsize}
\end{scriptsize} 		
  \end{overpic} 
\end{center} 
\caption{On the left (resp.\ right) side the shaky saddle realizations $G(\Vkt V')$
(resp.\ $G(\Vkt V'')$) of the design FH$_1$ is illustrated, which is passed during the snap between $G(\Vkt V_1)$ and 
$G(\Vkt V_2)$ (resp.\  $G(\Vkt V_3)$). 
An animation of the snapping behavior can be downloaded from \cite{data2021}.  
}\label{fig:2_FH}	
\end{figure}

\subsection{Four-horn}\label{sec:FH}

The original four-horn (FH) was introduced by Casper Schwabe at the {\it Ph\"anomena} exposition 1984 in Z\"urich, Swizerland 
(cf.\ Fig.\ \ref{fig:1_FH}).  
From the combinatorial point of view the FH equals a SD with pentagonal equatorial polygons. In contrast to a SD, a FH does not consist of congruent equilateral face-triangles but of 
congruent isosceles ones where $\alpha$ denotes the angle enclosed by the base of length $b>0$ and the leg of length $a>0$. 
Under consideration of the two-fold reflexion-symmetry with respect to two orthogonal planes, we can 
insert a Cartesian frame in such a way, that the vertices, which are noted according to Fig.\ \ref{fig:1_FH}, are coordinatized as follows: 
\begin{align}
A_1&=(x_1,y_1,0)^T &\quad \overline{A}_1&=(-x_1,y_1,0)^T &\quad B_1&=(0,y_2,0)^T &\quad C_1&=(x_3,y_3,0)^T &\quad \overline{C}_1&=(-x_3,y_3,0)^T \\
A_2&=(0,u_1,v_1)^T &\quad \overline{A}_2&=(0,u_1,-v_1)^T &\quad B_2&=(0,u_2,0)^T &\quad C_2&=(0,u_3,v_3)^T &\quad \overline{C}_2&=(0,u_3,-v_3)^T 
\end{align}
As in Appendix \ref{sec:SD} we can assume without loss of generality that $u_3=-y_3$ holds. 

It is well-known that the FH can snap out of the symmetric realization $G(\Vkt V_1)$
into one of the two flat realizations $G(\Vkt V_2)$ and $G(\Vkt V_3)$, respectively, which are both shaky due to their planarity.  
According to \cite{schwabe} these three undeformed realizations, which are displayed in Fig.\ \ref{fig:1_FH}, 
exist for all choices of $a$ and $b$ with $2a>b$. 

As done in \cite{schwabe} we will distinguish three different designs FH$_i$ $i=1,2,3$ of four-horns, which differ in the lengths of the leg $a_i$ and base $b_i$ with 
\begin{align}
a_1&=3\sqrt{2} + 6 - \frac{3}{2}\sqrt{20 + 14\sqrt{2}} 	&\quad  a_2&=6 - 3\sqrt{3}  &\quad a_3&=6\sqrt{3} + 12 -\frac{(9\sqrt{3} + 15)\sqrt{2}}{2} \\ 
b_1&=3\sqrt{20 + 14\sqrt{2}} -6\sqrt{2} - 9  						&\quad  b_2&= 6\sqrt{3}-9  	&\quad b_3&=(9\sqrt{3} + 15)\sqrt{2} - 12\sqrt{3} - 21
\end{align}
For these values, which result in an average edge length of 1, 
we get the angles $\alpha_1=22.5^{\circ}$ (the original design of Schwabe), $\alpha_2=30^{\circ}$ and $\alpha_3=15^{\circ}$, respectively. 

How the coordinates of the vertices can be computed for the two flat realizations $G(\Vkt V_2)$ and $G(\Vkt V_3)$ and the symmetric realization $G(\Vkt V_1)$ of these three designs FH$_i$ can be looked up in \cite{schwabe}. We only give the numerical values of these realizations in the Tables \ref{tab:FH1}--\ref{tab:FH3}.

\subsubsection{Isostaticity and shakiness}

The FH is isostatic for the same reasons as pointed out in \ref{sec:isoshaky}. 
Moreover, we can also determine the algebraic condition of shakiness in an analogous way, which yields:
\begin{equation}
x_1x_3^2v_1v_3^2
\underbrace{v_3(2x_1y_3 - x_3y_1 - x_3y_3)}_{\text{copl}(C_2,\overline{C}_2,A_1,C_1)}
\underbrace{x_3(u_1v_3 + 2v_1y_3 - v_3y_3)}_{\text{copl}(C_1,\overline{C}_1,A_2,C_2)}S=0
\end{equation}
where $x_1=0$ means that the triangles $(A_1,B_1,C_2)$ and $(A_1,B_1,\overline{C}_2)$ coincide with $(\overline{A}_1,B_1,C_2)$  and $(\overline{A}_1,B_1,\overline{C}_2)$, respectively.
The same holds for the condition $v_1=0$ by swapping the indices $1$ and $2$ for the above given triangles.
For the condition $x_3=0$ or $v_3=0$ two out of the four horns are in a flat configuration. 
Beside these geometric simple cases of shakiness we also have the factor\footnote{It can be downloaded from \cite{data2021}.}
$S$, which denotes an algebraic expression with 110 terms and a total degree of 9. 
Again $S$ is only quadratic with respect to the two non-zero coordinates of $A_i$ and $\overline{A}_i$, respectively, for $i=1,2$.

\subsubsection{Interpretation as a bar-joint structure}

For the three designs FH$_i$ for $i=1,2,3$ we compute similar to the SD example  given in Appendix
 \ref{sec:SD} the shaky saddle realizations 
$G(\Vkt V')$ and  $G(\Vkt V'')$, respectively, which are displayed in Fig.\ \ref{fig:2_FH} for FH$_1$. 
The numerical values of these realizations are also given in the Tables \ref{tab:FH1}--\ref{tab:FH3}. 
For all three designs the relation  $s(\Vkt L)=s(\Vkt V_i)=\varsigma(\Vkt V_1)$ holds true for $i=1,2,3$ (due to Theorem \ref{thm:ident} under consideration of
Corollary \ref{rep:shaky}) as well as $\varsigma(\Vkt L)=\varsigma(\Vkt V_2)=\varsigma(\Vkt V_3)=0$.
Moreover, we calculated for all three designs FH$_i$ the additional values $e_{min}$,  $\Delta L_{\diameter}^{abs}$, $\Delta L_{\diameter}^{rel}$, $\Delta L_{max}^{abs}$  and $\Delta L_{max}^{rel}$ as in  the  case of the SD. For a better comparison they are arranged in the Tables \ref{tab:FH_data1} and \ref{tab:FH_data2},  respectively. \bigskip

\begin{table}
\begin{center}
\begin{footnotesize}
\begin{tabular}[h]{c|ccccc}
 & \# tracked paths & \# $\mathcal{R}$ & \# $\mathcal{S}$ & $s(\Vkt L)=\varsigma(\Vkt V_1)$  & $e_{min}$ 
\\ \hline\hline
FH$_1$ & $19\,683$ & $924$ & $863$ & $1.753810068479\cdot 10^{-8}$  & $2.503636587824\cdot 10^{-3}$ 
\\ \hline
FH$_2$ & $19\,683$ & $917$ & $819$ & $2.035395987407\cdot 10^{-7}$  & $1.663070753397\cdot 10^{-2}$ 
\\ \hline
FH$_3$ & $19\,683$ & $923$ & $897$ & $9.864008781699\cdot 10^{-11}$ & $1.944647875494\cdot 10^{-4}$ 
\\ \hline
\end{tabular}
\end{footnotesize}
\end{center}
\caption{Computational data for the three designs FH$_1$, FH$_2$ and FH$_3$. Note that the computation of the set $\mathcal{R}$ was done by a 
total degree homotopy using Bertini.}\label{tab:FH_data1}
\end{table}

\begin{table}
\begin{center}
\begin{footnotesize}
\begin{tabular}[h]{c|cccc}
 &  $\Delta L_{\diameter}^{abs}$ & $\Delta L_{\diameter}^{rel}$ & $\Delta L_{max}^{abs}$  & $\Delta L_{max}^{rel}$   \\ \hline\hline
FH$_1$ & $1.755195468044\cdot 10^{-4}$ & $1.684100667119\cdot 10^{-4}$ & $2.932163649725\cdot 10^{-4}$ & $2.318551827789\cdot 10^{-4}$\\ \hline
FH$_2$ & $1.219270735675\cdot 10^{-3}$ & $1.174975973629\cdot 10^{-3}$ & $2.061810888944\cdot 10^{-3}$ & $1.830418324804\cdot 10^{-3}$\\ \hline
FH$_3$ &  $1.315683847557\cdot 10^{-5}$ & $1.257398895852\cdot 10^{-5}$ & $2.221096435873\cdot 10^{-5}$ & $1.598717509714\cdot 10^{-5}$\\ \hline
\end{tabular}
\end{footnotesize}
\end{center}
\caption{Continuation of Table \ref{tab:FH_data1}.}\label{tab:FH_data2}
\end{table}

\noindent
{\bf Comparison with the  method  presented in \cite{almost}:} 
According to \cite{private} the minimum bound $e_{min}^*$ of FH$_1$'s realization $G(\Vkt V_1)$ equals $4.0458\cdot 10^{-4}$, 
which is approximately 1/6 of the true barrier $e_{min}$ (cf.\ Table \ref{tab:FH_data1}). 
Moreover, $e_{min}^*$ implies an approximation of the absolute length $\Delta L_{\diameter}^*=4.129227333\cdot 10^{-5}$ an edge must change in average, which is about $23\%$ of the value  $\Delta L_{\diameter}^{abs}$ (cf.\ Table \ref{tab:FH_data1}).

\begin{rmk}
For the flat realizations the method of  \cite{almost} does not work, as they are not pre-stressed stable. 
Therefore we cannot compare the method of  \cite{almost} with the values 
obtained by our method regarding $G(\Vkt V_2)$ and $G(\Vkt V_3)$, respectively. \hfill $\diamond$
\end{rmk}

\noindent
{\bf Comparison with the results obtained in \cite{schwabe}:}             
The authors of \cite{schwabe} sliced the four-horn along the polylines 
$C_1C_2\overline{C}_1$ and $C_1\overline{C}_2C_1$ with exception of the points $C_1$ and $\overline{C}_1$. In this way they
get two two-horns, which are linked over the points $C_1$ and $\overline{C}_1$. 
Maintaining the 2-fold reflexion-symmetry, the resulting structure has a one-parametric mobility. 
Apart from the configurations $G(\Vkt V_1)$, $G(\Vkt V_2)$ and $G(\Vkt V_3)$ the points on both two-horns, 
which correspond to the 
point $C_2$ do not coincide. This mismatch of points is measured by the relative error in the z-coordinate. 
The maximum of this relative error during the flexion of the two two-horns between the two flat configurations 
equals the index 
given in  \cite{schwabe}. Therefore this index is also of extrinsic nature, but the authors of \cite{schwabe}
also noted an empirical rule of thumb without further explanation, which reads as $\delta^*=0.018\cdot(\alpha/10)^4\%$
where $\alpha$ has to be inserted in degree. This formula only depends on the intrinsic geometry of the 
framework. 

For our considered values $\alpha_i$ for $i=1,2,3$ we get $\delta_1^*=0.4613203125\%$, $\delta_2^*=1.458\%$ and 
$\delta_3^*=0.091125\%$. As this index $\delta^*$ cannot be compared one-to-one with any of our given values, 
we can evaluate the index $\delta^*$ by considering the relation $\delta_1^*:\delta_2^*:\delta_3^*$. 
It can easily be seen that this relation does not go along with the corresponding relation of any of the 
values  $s(\Vkt L)$, $e_{min}$,  $\Delta L_{\diameter}^{abs}$, $\Delta L_{\diameter}^{rel}$, 
$\Delta L_{max}^{abs}$  and $\Delta L_{max}^{rel}$, respectively, given in Tables \ref{tab:FH_data1} and \ref{tab:FH_data2}.

\subsubsection{Interpretation as a panel-hinge structure}
The same study can also be done by considering the FH as a polyhedral surface composed of 
triangular plates with a Poisson ratio of $\nu=1/2$.
We track for each of the three designs $19\,683$ paths of a total degree homotopy performed with Bertini. The computations end up in 
$1\,259$ real solutions for FH$_1$ ($1\,457$ for FH$_2$ and $1\,324$ for FH$_3$). 
After reduction to the set $\mathcal{S}$ we remain with $1\,242$ realizations for FH$_1$ ($1\,360$ for FH$_2$ and $1\,238$ for FH$_3$). 
Also for the interpretation as a panel-hinge structure the relations $\varsigma(\Vkt L)=\varsigma(\Vkt V_2)=\varsigma(\Vkt V_3)=0$ 
and $s(\Vkt L)=s(\Vkt V_i)=\varsigma(\Vkt V_1)$ ($i=1,2,3$) hold true 
for all three designs. 
The corresponding value equals $1.748173013469\cdot 10^{-6}$ for FH$_1$
($2.340885199965\cdot 10^{-5}$ for FH$_2$ and $6.288380657092\cdot 10^{-8}$ for FH$_3$).

\begin{table}
\begin{center}
\begin{footnotesize}
\begin{tabular}[h]{c|cccc}
 & $\Vkt V_1$ & $\Vkt V_2$ & $\Vkt V'$ (bar-joint) & $\Vkt V'$ (panel-hinge) \\ \hline\hline
 $x_1$ & -0.439833121345 	& -0.551313194956		& -0.514676938265  	& -0.513910947926  \\ \hline
 $x_3$ & -0.402578359944 	& -0.551313194956 	& -0.496123528337  	& -0.495616858966  \\ \hline
 $y_1$ & -1.045126760122 	& -1.055331194900 	& -1.049041632768 	& -1.047987400977  \\ \hline
 $y_2$ & -0.401357155967 	& -0.504017999944 	& -0.463165051665 	& -0.461829297096  \\ \hline
 $y_3$ & -0.266342733180 	& -0.275656597478 	& -0.269426558812 	& -0.269389463808  \\ \hline
 $u_1$ &  1.045126760122 	&  1.055331194900		&  1.047880696191		&  1.048797729978  \\ \hline
 $u_2$ &  0.401357155968 	&  0.275656597478 	&  0.331798913977 	&  0.332716034007  \\ \hline
 $v_1$ &  0.439833121346 	&  0								&  0.308030377883 	&  0.310260240907  \\ \hline
 $v_3$ &  0.402578359945 	&  0								&  0.267213277407 	&  0.268392688481  \\ \hline
\end{tabular}
\end{footnotesize}
\end{center}
\caption{
Coordinates of the undeformed realizations $G(\Vkt V_1)$ and  $G(\Vkt V_2)$ of the design FH$_1$ and of the 
shaky saddle realization  $G(\Vkt V')$ of FH$_1$ with respect to the two different interpretations.  
The coordinates of $G(\Vkt V_3)$ and $G(\Vkt V'')$ can be obtained from $G(\Vkt V_2)$ and $G(\Vkt V')$
by the following exchange of coordinate entries:
$x_j\leftrightarrow -v_j$  for $j=1,3$ and  $y_i\leftrightarrow -u_i$ for $i=1,2,3$.
}\label{tab:FH1}
\end{table}

\begin{table}
\begin{center}
\begin{footnotesize}
\begin{tabular}[h]{c|cccc}
 & $\Vkt V_1$ & $\Vkt V_2$ & $\Vkt V'$ (bar-joint) & $\Vkt V'$ (panel-hinge) \\ \hline\hline
 $x_1$ & 	-0.610560396069 	& -0.696152422706 	& -0.674892191647  & -0.673709307095  \\ \hline
 $x_3$ &  -0.514152040259 	& -0.696152422706 	& -0.629718504095  & -0.630052932256  \\ \hline
 $y_1$ &  -0.969412109993 	& -1.004809471616 	& -0.984812341395  & -0.981100971148  \\ \hline
 $y_2$ &  -0.446547949553 	& -0.602885682969 	& -0.546021679456  & -0.541654138725  \\ \hline
 $y_3$ &  -0.171366775159 	& -0.200961894323 	& -0.181083746571  & -0.180980608226  \\ \hline
 $u_1$ &   0.969412109993 	&  1.004809471616 	&  0.975854359558  &  0.978589555099  \\ \hline
 $u_2$ &   0.446547949553 	&  0.200961894323 	&  0.316170888257  &  0.318169735899  \\ \hline
 $v_1$ &   0.610560396069 	&  0								&  0.457391731880  &  0.463293753345  \\ \hline
 $v_3$ &   0.514152040259  	&  0								&  0.344387159440  &  0.344641218655  \\ \hline
\end{tabular}
\end{footnotesize}
\end{center}
\caption{The analogous table to Table \ref{tab:FH1} but with respect to the design FH$_2$.}\label{tab:FH2}
\end{table}

\begin{table}
\begin{center}
\begin{footnotesize}
\begin{tabular}[h]{c|cccc}
 & $\Vkt V_1$ & $\Vkt V_2$ & $\Vkt V'$ (bar-joint) & $\Vkt V'$ (panel-hinge) \\ \hline\hline
 $x_1$ & 	-0.284308975844 	& -0.381499642545 	& -0.346332965123  	& -0.345941194845  \\ \hline
 $x_3$ &  -0.274123668705 	& -0.381499642545 	& -0.341068732408 	& -0.340681988409  \\ \hline
 $y_1$ &  -1.091519316228 	& -1.093387667069 	& -1.092180085853 	& -1.092027657175  \\ \hline
 $y_2$ &  -0.383468247941 	& -0.432610903111 	& -0.412298335661	 	& -0.412002561558  \\ \hline
 $y_3$ &  -0.328588016197 	& -0.330388381978 	& -0.329187289575  	& -0.329179948495  \\ \hline
 $u_1$ &   1.091519316228 	&  1.093387667069 	&  1.092099843041  	&  1.092235340815  \\ \hline
 $u_2$ &   0.383468247944 	&  0.330388381978 	&  0.353323819879 	&  0.353582747289  \\ \hline
 $v_1$ &   0.284308975853 	&  0								&  0.190684818526 	&  0.191590754090  \\ \hline
 $v_3$ &   0.274123668714 	&  0								&  0.179953619106 	&  0.180744480102  \\ \hline
\end{tabular}
\end{footnotesize}
\end{center}
\caption{The analogous table to Table \ref{tab:FH1} but with respect to the design FH$_3$.}\label{tab:FH3}
\end{table}

\end{document}